\documentclass[12pt]{article}
\usepackage[small]{titlesec}
\usepackage{amssymb,amsmath,amsthm}
\usepackage{enumerate}
\usepackage{hyperref}
\usepackage{cite}
\usepackage[mathscr]{euscript}
\usepackage[letterpaper,hmargin=2.70cm,vmargin=3.2cm]{geometry}
\usepackage{setspace}
\usepackage{graphicx}
\newtheorem{theorem}{Theorem}[section]
\newtheorem{proposition}[theorem]{Proposition}
\newtheorem{lemma}[theorem]{Lemma}
\newtheorem{corollary}[theorem]{Corollary}
\theoremstyle{definition}

\theoremstyle{remark}
\newtheorem{remark}{Remark}[section]
\newtheorem*{example}{Example}

\newcommand{\abs}[1]{\left| #1 \right|}
\newcommand{\norm}[1]{\left\| #1 \right\|}
\newcommand{\I}{{\rm i}}
\newcommand{\inner}[2]{\left\langle#1,#2\right\rangle}

\newcommand{\dual}[2]{\left[#1,#2\right]}
\newcommand{\cD}{\mathcal{D}}
\newcommand{\cH}{\mathcal{H}}
\newcommand{\cB}{\mathcal{B}}

\newcommand{\R}{{\mathbb R}}
\newcommand{\C}{{\mathbb C}}
\newcommand{\N}{{\mathbb N}}
\newcommand{\Z}{{\mathbb Z}}
\newcommand{\cc}[1]{\overline{#1}}
\newcommand{\ournewclass}{\mathscr{S}(\mathcal{H})}
\newcommand{\nentireclass}[1]{\mathscr{E}_{#1}(\mathcal{H})}
\newcommand{\mb}{\boldsymbol}
\newcommand{\pre}{\mbox{\scriptsize\rm pre}}

\renewcommand\hat{\widehat}
\DeclareMathOperator{\im}{im}
\DeclareMathOperator{\dom}{dom}
\DeclareMathOperator{\Ker}{ker}

\DeclareMathOperator{\ran}{ran}
\DeclareMathOperator{\Sp}{spec}
\DeclareMathOperator{\Span}{span}

\DeclareMathOperator{\assoc}{assoc}

\begin{document}
\begin{titlepage}
\title{The class of $n$-entire operators% and its spectral characterization%
\footnotetext{%
Mathematics Subject Classification(2000):
% 41A05, % Interpolation
46E22, % Hilbert Spaces with reproducing kernels (functional Hilbert
       % spaces, including de Branges-Rovnyak and other structured spaces)
% 47A55, % Perturbation theory
47A25, %Spectral sets
47B25, % Symmetric and selfadjoint operators (unbounded)
47N99.} % Applications of operator theory
	% Sampling theory
\footnotetext{%
PACS Numbers:
         02.30.-f, % Function theory, analysis
	  02.30.Tb, % Operator theory
	  02.70.Hm} % Spectral methods
\footnotetext{%
Keywords: symmetric operators, entire operators, de Branges spaces, spectral
	analysis.}
\\[6mm]}
\author{
\textbf{Luis O. Silva}\thanks{Partially supported by CONACYT (M{\'e}xico)
	through grant CB-2008-01-99100}
\\
%% ----- Institution ----------
\small Departamento de F\'{i}sica Matem\'{a}tica\\[-1.6mm]
\small Instituto de Investigaciones en Matem\'{a}ticas Aplicadas y
	en Sistemas\\[-1.6mm]
\small Universidad Nacional Aut\'{o}noma de M\'{e}xico\\[-1.6mm]
\small C.P. 04510, M\'{e}xico D.F.\\[-1.6mm]
\small \texttt{silva@leibniz.iimas.unam.mx}
\\[4mm]
\textbf{Julio H. Toloza}\thanks{Partially supported by CONICET (Argentina)
	through grant PIP 112-200801-01741}
\\
%% ----- Institution ----------
\small CONICET\\[-1.6mm]
\small Centro de Investigaci\'{o}n en Inform\'{a}tica para la
	Ingenier\'{i}a\\[-1.6mm]
\small Universidad Tecnol\'{o}gica Nacional --
	 Facultad Regional C\'{o}rdoba\\[-1.6mm]
\small Maestro L\'{o}pez esq.\ Cruz Roja Argentina\\[-1.6mm]
\small X5016ZAA C\'{o}rdoba, Argentina\\[-1.6mm]
\small \texttt{jtoloza@scdt.frc.utn.edu.ar}}

\date{}
\maketitle
\begin{center}
\begin{minipage}{5in}
  \centerline{{\bf Abstract}} \bigskip
  We introduce a classification of simple, regular, closed symmetric
  operators with deficiency indices $(1,1)$ according to a geometric
  criterion that extends the classical notions of entire operators and
  entire operators in the generalized sense due to M. G. Krein.  We
  show that these classes of operators have several distinctive
  properties, some of them related to the spectra of their canonical
  selfadjoint extensions.  In particular, we provide necessary and
  sufficient conditions on the spectra of two canonical selfadjoint
  extensions of an operator for it to belong to one of our
  classes. Our discussion is based on some recent results in the
  theory of de Branges spaces.
\end{minipage}
\end{center}
\thispagestyle{empty}
\end{titlepage}

%%%%%%%%%%%%%%%%%%%%%%%%%%%%%%%%%%%%%%%%%%%%%%%%%%%%%%%%%%%%%%%%%%%%%%%%%%%%%%%

\section{Introduction}
\label{sec:introduction}

Let $\ournewclass$ be the class of regular, closed symmetric operators
on a separable Hilbert space $\cH$, whose deficiency indices are
$(1,1)$ (see details in section 2). It is well known that operators of
this class share a number of distinctive properties. For instance, all
the canonical selfadjoint extensions of a given operator
$A\in\ournewclass$ have simple discrete spectra, pairwise interlaced,
whose union is the real line. Also, associated to each
$A\in\ournewclass$, there exists a unitary transformation that maps
$\cH$ onto a de Branges space (a special kind of Hilbert space of
entire functions \cite{debranges}), on which $A$ is unitarily
transformed into the multiplication operator by the independent
variable \cite{II}. These facts, among others, have been exploited in
more or less explicit form in the study of diverse questions of
interest in mathematical physics, like boundary-value and inverse
problems of canonical systems \cite{hassi2,winkler}, the spectral
analysis of Krein strings \cite{krein-string}, inverse spectral
problems of one-dimensional Schr{\"o}dinger operators \cite{remling}
(see also \cite{eckhardt-preprint} for recent developments in the
case of strongly singular potentials), analysis of minimum uncertainty
for quantum observables \cite{martin-kempf1}, and some related problems
in quantum gravity \cite{kempf1,kempf2,kempf3}, to mention a few of
them. Besides applications in mathematical physics, operators in
$\ournewclass$ has been used in some aspects of signal processing and
analytical sampling theory (see for instance \cite{I}).

In this paper we introduce a classification of operators within
$\ournewclass$.  Namely, for every given $n\in\Z^+=\{0,1,\dots\}$, we
consider those operators $A\in\nentireclass{n}\subset\ournewclass$ for
which one can find $n+1$ vectors $\mu_0,\dots,\mu_n\in\cH$ such that
\begin{equation}
\label{eq:n-entireness-algebraic}
	\cH = \ran(A-zI)\dot{+}\Span\{\mu_0+z\mu_1+\cdots+z^n\mu_n\},
	\mbox{ for all } z\in\C.
\end{equation}

The aim of this paper is to discuss a number of properties that are
common to all operators within each class $\nentireclass{n}$, some of
them related to the spectra of their canonical selfadjoint extensions,
some others connected to their associated de Branges spaces.  It will
be shown that our classification carries out a refinement in the
characterization of some (but not all) operators in $\ournewclass$.

Among the operators that obey \eqref{eq:n-entireness-algebraic} are the
entire operators as well as the entire operators in the generalized
sense.  These classes of operators, which include operators frequently
appearing in mathematical physics, were originally concocted by
M. G. Krein as a tool for treating in a unified way several classical
problems in analysis \cite{krein1,krein2,krein3,krein4}. A detailed
review of entire operators and their many remarkable properties is
\cite{gorbachuk}.  Because of this connection with entire operators,
the class $\nentireclass{n}$ will henceforth be referred to as the
class of $n$-entire operators.

Let us describe briefly the relation between Krein's definitions and
ours here, referring the details to Section \ref{sec:preliminaries}.
In what follows let $\inner{\cdot}{\cdot}$ denote the inner product on
$\cH$, assumed antilinear in its first argument.
We recall that a simple, regular, closed symmetric operator $A$,
{\em densely defined} on $\cH$, with deficiency indices $(1,1)$, is
entire (according to Krein) if there exists $\mu\in\cH$ such that
$\cH=\ran(A-zI)\dot{+}\Span\{\mu\}$ for all $z\in\cH$. Equivalently, $A$
is entire if $\inner{\xi(\cc{z})}{\mu}$ is a zero-free entire function,
where $\xi(z)$ is a certain vector-valued zero-free entire function such that
$\xi(z)\in \Ker(A^*-zI)$ (for details see \cite{gorbachuk,II}). The operator
$A$ is entire in the generalized sense if there exists $\mu\in\cH_{-}$
such that $\dual{\xi(\cc{z})}{\mu}$ is a zero-free entire function,
where $\cH_{-}$ is the dual of $\cH_{+}:=\dom(A^*)$ equipped with the graph
norm, and $\dual{\cdot}{\cdot}$ denotes the associated duality bracket.
Clearly, an operator entire according to Krein's definition is $0$-entire.
It is a bit less apparent that an operator entire in the generalized
sense is indeed $1$-entire. To see this we observe that, as a direct
consequence of \cite[Proposition 5.1]{II}, given
$\mu\in\cH_{-1}\setminus\cH$ one can find $\mu_0,\mu_1\in\cH$ such that
\[
	[\xi(\cc{z}),\mu]
		= \inner{\xi(\cc{z})}{\mu_0}+z\inner{\xi(\cc{z})}{\mu_1}
\]
for all $z\in\C$,
hence reducing Krein's to our definition with $n=1$.
It is worth remarking here that the class $\ournewclass$ includes operators
with non-dense domain. That is, our classes of $0$-entire and $1$-entire
operators are themselves larger that the corresponding Krein's classes.

As first discussed in \cite{II} it is possible to determine whether an
operator is either entire or entire in the generalized sense by
conditions that rely exclusively on the distribution of the spectra of
selfadjoint extensions of the operator. This spectral characterization
was obtained on the basis of recent results in the theory of de
Branges spaces \cite{woracek},\cite{woracek2}.  One of the main
results of this paper is a generalization of this spectral
characterization to $n$-entire operators.

This paper is organized as follows. In Section \ref{sec:preliminaries}
we introduce the main concepts relevant to this work. We also present
here the mathematical background (including some new results) needed later.
In Section \ref{sec:main-results}
we discuss several characterizations for the classes of operators discussed
in this paper. Section \ref{sec:gelfand-triplets} is devoted to the construction
of a Gelfand triplet associated to $n$-entire operators, in an attempt to
recover the original way that Krein used to introduce the notion of operator
entire in the generalized sense. Finally, we draw some conclusions and point
out some ideas for further investigation in Section \ref{sec:conclusions}.

%%%%%%%%%%%%%%%%%%%%%%%%%%%%%%%%%%%%%%%%%%%%%%%%%%%%%%%%%%%%%%%%%%%%%%%%%%%%%%%

\section{Symmetric operators and de Branges spaces}
\label{sec:preliminaries}

In this section we lay out the notation and introduce some of the main
objects to be considered in this work. The first part of the section
deals with symmetric operators. The operator classes which will be
discussed in this work are defined here. The second part is devoted to
the theory of de Branges spaces. Finally, the last part of this
section deals with the construction of the functional model for the
operators considered in the first part. The functional model serves as
a bridge that relates every operator in $\ournewclass$ to a certain de
Branges space.

%%%%%%%%%%%%%%%%%%%%%%%%%%%%%

\subsection{On symmetric operators with not necessarily dense domain}
\label{subsec:symmetric}

Let $\cH$ be a Hilbert space whose inner product
$\inner{\cdot}{\cdot}$ is assumed antilinear in its first argument. In
this space we consider a closed, symmetric linear operator $A$ with
deficiency indices $(1,1)$. It is not presumed that its domain is
dense in $\cH$, therefore one should deal with the case when the
adjoint of $A$ is a closed linear relation. Recall that a closed
linear relation in $\cH$ is a subspace of $\cH\oplus\cH$ and,
therefore, closed operators are closed linear relations when they are
identified with their graphs. Thus, in general,
\begin{equation}
\label{eq:adjoint-def}
A^* := \left\{(\eta,\omega)\in\cH\oplus\cH :
		\inner{\eta}{A\varphi}=\inner{\omega}{\varphi}
		\mbox{ for all }\varphi\in\dom(A)\right\}.
\end{equation}
Whenever the orthogonal complement of $\dom(A)$ is trivial, the set
\[
A^*(0):=\{\omega\in\cH:(0,\omega)\in A^*\}
\]
is also trivial, i.e. $A^*(0)=\{0\}$, so $A^*$ is an operator;
otherwise $A^*$ is a multivalued closed linear relation.

For $z\in\C$ one has
\begin{equation}
\label{eq:adjoint-shifted}
A^*-zI := \left\{(\eta,\omega-z\eta)\in\cH\oplus\cH :
		(\eta,\omega)\in A^*\right\}\,,
\end{equation}
so accordingly
\begin{equation}
\label{eq:ker-adjoint}
\Ker(A^*-zI) := \left\{\eta\in\cH : (\eta,0)\in A^*-zI \right\}.
\end{equation}
Therefore, on the basis of the decomposition
\begin{equation}
  \label{eq:space-decomposition}
  \cH=\ran(A-zI)\oplus\ker(A^*-\cc{z}I)\,,
\end{equation}
which holds independently of the fact that $A$ is or not densely defined
\cite[proposition 3.31]{arens}, our assumption on the deficiency indices
implies
$\dim\Ker(A^*-zI)=1$ for all $z\in\C\setminus\R$.  Moreover, since
\begin{equation*}
A^*(0) = \left\{\omega\in\cH : \inner{\omega}{\psi} = 0
		\mbox{ for all }\psi\in\dom(A)\right\},
\end{equation*}
it is obvious that $A^*(0) = \dom(A)^\perp$.

In this work we deal not only with symmetric operators but with their
canonical selfadjoint extensions. A canonical selfadjoint extension of
a given symmetric operator is a selfadjoint extension within the
original space $\cH$, i.\,e., a selfadjoint extension of $A$
being a restriction of $A^*$. If $A$ turns out not to be densely
defined, then a canonical selfadjoint extension $A_\gamma$ of $A$ is a
subspace of $A^*$ that extends the graph of $A$ and that satisfies
$A_\gamma^*=A_\gamma$ (as subsets of $\cH\oplus\cH$).

The following proposition concerns non-densely defined symmetric
operators. It shows that the condition for the deficiency indices to be
$(1,1)$ implies that these operators are not quite dissimilar to the
densely defined ones. A proof of this proposition follows
from \cite[section 1, lemma 2.2 and theorem 2.4]{hassi1} (see
\cite[proposition 5.4]{hassi1} and the comment below it).

\begin{proposition}
\label{prop:misc-about-symm-operators}
Let $A$ be a closed, non-densely defined, symmetric operator in a
Hilbert space.
If $A$ has deficiency indices $(1,1)$, then
\begin{enumerate}[{(i)}]
\item the codimension of $\dom(A)$ equals one.
\item all except one of the canonical selfadjoint extensions of $A$
	are operators.
\end{enumerate}
\end{proposition}

Let us now bring up a simple result which does not depend on whether
the operator is densely defined or not. The proof of it can be found
in \cite[section 1]{hassi1} for the nondensely defined case and in
\cite[section 2.1]{gorbachuk} for the densely defined one. Before
stating it, we remind the reader that the spectrum of a closed linear
relation $B$ in $\cH$ is the complement of the set of all $z\in\C$
such that $(B-zI)^{-1}$ is a bounded operator defined on all $\cH$.
Moreover, $\Sp(B)\subset\R$ when $B$ is a selfadjoint linear relation
\cite{dijksma}.

\begin{proposition}
  \label{prop:misc-about-symm-operator2}
  Let $A$ be a closed, symmetric operator in a Hilbert space. If
  $A_\gamma$ is a canonical selfadjoint extension of $A$, then the
  operator
	\begin{equation*}%\label{eq:gen-cayley-transform}
	I + (z-w)(A_\gamma-zI)^{-1},\quad
        z\in\C\setminus\Sp(A_\gamma),\quad
        w\in\C
	\end{equation*}
	maps $\Ker(A^*-wI)$ injectively onto $\Ker(A^*-zI)$.
\end{proposition}

The operator given in this proposition is the generalized Cayley
transform and we use it to define a function taking values in
$\ker(A^*-z I)$ as follows
\begin{equation}\label{eq:mapping-in-kernel}
\psi(z):= \left[I + (z-w_0)(A_\gamma-zI)^{-1}\right]\psi_{w_0},
\end{equation}
for given $\psi_{w_0}\in\Ker(A^*-w_0I)$ and
$w_0\in\C\setminus\R$. Clearly, $\psi(\cdot)$ is an analytic function
in the upper and lower half-planes because of the analytic properties
of the resolvent.
Obviously, $\psi(w_0)=\psi_{w_0}$. Moreover, a computation involving the
resolvent identity yields
\begin{equation}\label{eq:identity-between-kernels}
\psi(z) = \left[I + (z-v)(A_\gamma-zI)^{-1}\right]\psi(v),
\end{equation}
for any pair $z,v\in\C\setminus\R$. This identity will be used later on.

For the sake of completeness, and also for future reference, we recall
the notion of simplicity of a closed symmetric nonselfadjoint
operator. A closed symmetric nonselfadjoint operator is said to be
simple (or completely nonselfadjoint) if it is not a nontrivial
orthogonal sum of a symmetric and a selfadjoint operators. Since an
invariant subspace of a symmetric operator is a subspace reducing that
operator \cite[theorem 4.6.1]{birman}, a symmetric operator $A$ is
simple when there is not a nontrivial invariant subspace of $A$ on
which $A$ is selfadjoint.

By \cite[proposition 1.1]{langer-textorius} (see
\cite[theorem 1.2.1]{gorbachuk} for the densely defined case), a
necessary and sufficient condition for the symmetric nonselfadjoint
operator $A$ to be simple is
\begin{equation}
  \label{eq:simplicity}
  \bigcap_{z\in\C\setminus\R}\ran(A-zI)=\{0\}\,.
\end{equation}

Simplicity plays an important role in our further considerations. Here
we briefly discuss some of the distinctive features that a closed
symmetric operator with deficiency indices $(1,1)$ has when it is
simple. Consider the function $\psi(\cdot)$ given by
\eqref{eq:mapping-in-kernel} and take a sequence
$\{z_k\}_{k=1}^\infty$ with elements in
$\mathbb{C}\setminus\mathbb{R}$ having accumulation points in the
upper and lower half-planes. Suppose that there is
$\eta\in\cH$ such that $\inner{\eta}{\psi(z_k)}=0$ for all
$k\in\mathbb{N}$. This implies that $\inner{\eta}{\psi(z)}=0$ for
$z\in\mathbb{C}\setminus\mathbb{R}$ because of the analyticity of the
function $\inner{\eta}{\psi(\cdot)}$. Therefore, by
\eqref{eq:simplicity}, $\eta=0$. We have thus arrived at the conclusion
that simple, closed symmetric operators with deficiency indices
$(1,1)$ can exist only in {\em separable} Hilbert spaces.
From now on, the reader should assume that $\cH$ is separable.

Another property of simple, closed symmetric operators with
deficiency indices $(1,1)$ concerns their
commutativity with involutions and it is the content of the
next proposition. We say that an involution $J$ commutes with
a selfadjoint relation $B$ if
\[
J(B-zI)^{-1}\varphi = (B-\cc{z}I)^{-1}J\varphi,
\]
for every $\varphi\in\cH$ and $z\in\C\setminus\R$. If $B$ is moreover an
operator this is equivalent to the usual notion of commutativity, that is,
\[
J\dom(B)\subseteq\dom(B),\qquad JB\varphi = BJ\varphi
\]
for every $\varphi\in\dom(B)$.

\begin{proposition}
\label{prop:existence-of-commuting-involution}
  Let $A$ be a simple, closed symmetric operator with deficiency
  indices $(1,1)$. Then there exists an involution $J$ that commutes
  with all its canonical selfadjoint extensions.
\end{proposition}

\begin{proof}
  Choose a selfadjoint extension $A_\gamma$ and consider $\psi(z)$ as
  defined by \eqref{eq:mapping-in-kernel}. Recalling
  \eqref{eq:identity-between-kernels} along with the unitary character
  of the generalized Cayley transform, and applying the resolvent
  identity, one can verify that
\begin{equation}
\label{eq:pre-j}
\inner{\psi(\cc{z})}{\psi(\cc{v})} = \inner{\psi(v)}{\psi(z)}
\end{equation}
for every pair $z,v\in\C\setminus\R$.

Now define the action of $J$ on $\psi(z)$
($z\in\mathbb{C}\setminus\mathbb{R}$) by
the rule
\[
J\psi(z)=\psi(\cc{z}),
\]
and on the set $\cD$ of finite linear combinations of such elements by
\[
J\left(\sum_{n}c_n\psi(z_n)\right): = \sum_{n}\cc{c_n}\psi(\cc{z_n})\,,
\]
where the sequence $\{z_k\}_{k=1}^\infty$ is defined as in the
paragraph following \eqref{eq:simplicity}.
Then, on one hand, \eqref{eq:pre-j} implies that $J$ is an involution on
$\cD$ which can be extended to all $\cH$ because of the simplicity of $A$.
On the other hand, since by the resolvent identity
\[
(A_\gamma-wI)^{-1}\psi(z) = \frac{\psi(z)-\psi(w)}{z-w},
\]
one obtains the identity
\[
J(A_\gamma-wI)^{-1}\psi(z)=(A_\gamma-\cc{w}I)^{-1}J\psi(z)
\]
which by linearity holds on $\cD$ and in turn it extends to all
$\cH$.

So far we know that $J$ commutes with $A_\gamma$. By resorting to the
well-known resolvent formula due to Krein (see \cite[theorem 3.2]{hassi1}
for a generalized formulation), one immediately obtains
the commutativity of $J$ with all the selfadjoint extensions of $A$
within $\cH$.
\end{proof}

We now remind the reader the notion of regularity of a closed
operator. A closed operator $A$ in $\cH$ is regular if for every $z\in\C$
there exists $d_z>0$ such that
\begin{equation}
\label{eq:regular-point}
\norm{(A-zI)\psi}\ge d_z\norm{\psi},
\end{equation}
for all $\psi\in\dom(A)$. In other words, $A$ is regular if every
point of the complex plane is a point of regular type.

It is easy to see that a regular, closed symmetric operator is
necessarily simple, this is so because the regularity implies the lack of
spectral kernel. The converse statement is not true, however.

Let us define the operator class $\ournewclass$ as the set of all
regular, closed symmetric operators with deficiency indices $(1,1)$. By
what have just been said in the paragraph above all the operators in
$\ournewclass$ are simple. But regularity adds also further properties to
the class $\ournewclass$. Indeed, the combination of regularity and the
fact that the deficiency indices are $(1,1)$ leads to the following
proposition which extends to the whole class $\ournewclass$ well-known
facts for densely defined operators in $\ournewclass$.

\begin{proposition}
\label{prop:properties-of-new-class}
For $A\in\ournewclass$ the following assertions hold true:
\begin{enumerate}[{(i)}]
\item The spectrum of every canonical selfadjoint extension of $A$
	consists solely of isolated eigenvalues of multiplicity one.
\item Every real number is part of the spectrum of one, and only one,
	canonical selfadjoint extension of $A$.
\item The spectra of the canonical selfadjoint extensions of $A$ are
	pairwise interlaced.
\end{enumerate}
\end{proposition}
\begin{proof}
  We will prove (i) using similar ideas as in the proofs of
  propositions 3.1 and 3.2 of \cite{gorbachuk}, but taking into account
  that the operator is not necessarily densely defined.

  For $A\in\ournewclass$ and any $r\in\R$ consider the constant $d_r$
  of \eqref{eq:regular-point}. Thus, the symmetric operator
  $(A-rI)^{-1}$, defined on the subspace $\ran(A-rI)$, is such that
  $\norm{(A-rI)^{-1}}\le d_r^{-1}$. By \cite[theorem
  2]{krein-half-bounded}, there is a selfadjoint extension $B$ of
  $(A-rI)^{-1}$ defined on the whole space and such that $\norm{B}\le
  d_r^{-1}$. Now, $B^{-1}$ is a selfadjoint extension of $A-rI$ and
  $\norm{B^{-1}f}\ge d_r\norm{f}$ for any $f\in\dom(B^{-1})$, which
  implies that $(-d_r,d_r)\cap\Sp(B^{-1})=\emptyset$. By appropriately
  shifting $B^{-1}$ one obtains a selfadjoint extension of $A$ with no
  spectrum in the spectral lacuna $(r-d_r,r+d_r)$. Now, according to
  perturbation theory any selfadjoint extension of $A$ which is an
  operator has no points of the spectrum in this spectral lacuna other
  than one eigenvalue of multiplicity one. To prove (i) for operator
  extensions, consider any closed interval of $\R$, cover it
  with spectral lacunae, and take a finite subcover. Actually (i) also
  holds for the only selfadjoint multivalued relation in the case
  $\cc{\dom(A)}\ne\cH$. This follows from the simplicity of the
  operator selfadjoint extensions and \cite[equation 3.10]{hassi1}.

  Once (i) has been proven, the assertion (ii) and (iii) follow again from
  \cite[equation 3.10]{hassi1} and the properties of Herglotz
  meromorphic functions.
\end{proof}

We now turn to the discussion of the notion of entire operators and their
generalizations. A vector $\mu\in\cH$ is said to be a gauge for
(a given operator) $A\in\ournewclass$ if and only if
\begin{equation}\label{eq:gauge-definition}
\cH = \ran(A-z_0I)\dot{+}\Span\{\mu\}
\end{equation}
for some $z_0\in\mathbb{C}$, where $\dot{+}$ denotes the direct
sum. Once a gauge has been chosen, we look for the set of complex
numbers for which \eqref{eq:gauge-definition} fails to hold, viz.,
\begin{equation}
  \label{eq:special-set}
  \left\{z\in\mathbb{C}:\mu\perp\ker(A^*-\cc{z}I)\right\}.
\end{equation}
The set \eqref{eq:special-set} is at most an infinite countable set
with no finite accumulation points (see \cite[section 2]{II}). Moreover,
depending on the choice of the gauge $\mu$, the set
\eqref{eq:special-set} could be entirely contained in $\R$
\cite[lemma 2.1]{II} or placed completely outside $\R$ \cite[theorem 2.2]{II}.

If the gauge $\mu$ can be chosen so that the set \eqref{eq:special-set}
is empty, then the gauge is said to be entire. In other words,
$\mu\in\cH$ is an entire gauge if and only if
\begin{equation}\label{eq:entire-gauge-definition}
\cH = \ran(A-zI)\dot{+}\Span\{\mu\},
\end{equation}
for all $z\in\mathbb{C}$.

Within $\ournewclass$, we single out the class $\nentireclass{0}$ of
operators for which there exists an entire gauge. The operators in
$\nentireclass{0}$ are called entire operators.  The densely defined
operators in $\nentireclass{0}$ were originally introduced by Krein in
the 1940's for the purpose of treating in a unified way several
classical problems in mathematical analysis
\cite{krein1,krein2,krein3,krein4}. It is worth remarking that the
extension of the concept of entire operators from the densely defined
ones to the not necessarily densely defined operators is completely
natural in the light of the investigations carried out by de Branges
on certain Hilbert spaces of entire functions in the 1960's. This will
become clear in subsection~\ref{subsec:functional-model}.

Krein's theory of entire operators is constructed on the basis of a
particular functional model for densely defined operators in the class
$\ournewclass$. This functional model was generalized in an abstract
way in \cite{strauss1,strauss2} to include operator classes broader than
$\ournewclass$. Subsection~\ref{subsec:functional-model} provides a
realization of the abstract construction of \cite{strauss1,strauss2} based on
the function given in \eqref{eq:mapping-in-kernel}. Basically, the idea
behind our functional model is to construct a function which
associates to any complex number $z$ a vector $\xi(z)\in\ker(A^*-z
I)$. By means of this function one says that $A$ is in
$\nentireclass{0}$ if and only if there exists a $\mu\in\cH$ such that for
all $z\in\C$
\begin{equation}
  \label{eq:entire-def-xi}
  \inner{\xi(\cc{z})}{\mu}\ne 0\,.
\end{equation}

Besides entire operators, Krein considered the so-called entire
operators in the generalized sense. These operators were studied by
\u{S}muljan (see for instance \cite{smuljan}) and their definition
is as follows. Take a densely defined operator $A\in\ournewclass$ and
consider the Hilbert space $\cH_+$ being the linear set $\dom(A^*)$
equipped with the graph norm. Let $\cH_-$ be the dual of $\cH_+$, that is,
the collection of $\cH_+$-continuous anti-linear functionals. Clearly,
$\cH_+\subset\cH\subset\cH_-$ (for the details see
Section \ref{sec:gelfand-triplets} below).  Then, $A$ is entire in the
generalized sense when there is a $\mu\in\cH_-\setminus\cH$ such that
for all $z\in\C$, one has, instead of \eqref{eq:entire-def-xi},
\begin{equation*}
  \dual{\xi(\cc{z})}{\mu}\ne 0\,,
\end{equation*}
where $\dual{\cdot}{\cdot}$ denotes the duality bracket between
$\cH_+$ and $\cH_-$.

It will be proven below (see proposition~\ref{prop:definition-makes-sense})
that a densely defined operator $A$ is entire in the generalized sense
if and only if there are vectors $\mu_0,\mu_1\in\cH$ such that
\begin{equation}
\label{eq:op-entire-generalized}
  \cH=\ran(A-zI)\dot{+}\Span\{\mu_0+z\mu_1\}.
\end{equation}
Clearly this definition makes sense whether or not the operator is densely
defined. This motivated us to single out the class $\nentireclass{1}$
of operators entire in the generalized sense
as the collection of operators in $\ournewclass$ that satisfies
\eqref{eq:op-entire-generalized}.

At this point it is clear that our definition of the classes
$\nentireclass{n}$, of operators in $\ournewclass$ that
fulfills \eqref{eq:n-entireness-algebraic} for a given $n\in\Z^+$, is
the natural generalization of the classes $\nentireclass{0}$ and
$\nentireclass{1}$. These classes are ordered in the following sense,
\[
\nentireclass{0}\subset\nentireclass{1}
	\subset\nentireclass{2}\subset\cdots\subset\ournewclass.
\]
However,
\[
\bigcup_{n\in\Z^+}\nentireclass{n}\subsetneq\ournewclass,
\]
as it will become clear in section~\ref{sec:main-results} and
illustrated by example~\ref{example:counter-ex}.

\begin{example}
  Here we construct densely and nondensely defined $0$-entire
  operators using Jacobi matrices. These matrices appear often in the
  mathematical physics literature not only because of the theoretical
  significance the corresponding operators have for being the discrete
  analogue of Sturm-Liouville operators, but also because they are
  used for modeling physical processes as in solid state physics
  within the so-called tight binding approximation \cite[chapter
  9]{MR0883643}, quantum optics \cite{tur}, and mechanics
  \cite[section 1.5 and part 2]{MR1711536}.

Consider the semi-infinite Jacobi matrix
\begin{equation}
  \label{eq:jm}
  \left(
  \begin{array}{ccccc}
  q_1 & b_1 & 0  &  0  &  \cdots \\[1mm]
  b_1 & q_2 & b_2 & 0 & \cdots \\[1mm]
  0  &  b_2  & q_3  & b_3 &  \\
  0 & 0 & b_3 & q_4 & \ddots\\
  \vdots & \vdots &  & \ddots & \ddots
  \end{array}\right),
\end{equation}
where $b_k>0$ and $q_k\in\mathbb{R}$ for $k\in\mathbb{N}$. Fix an
orthonormal basis $\{\delta_k\}_{k\in\mathbb{N}}$ in $\cH$. Let $B$ be
the operator in $\cH$ whose matrix representation with respect to
$\{\delta_k\}_{k\in\mathbb{N}}$ is (\ref{eq:jm}) (cf. \cite[section
47]{akhiezer2}).  We assume that $B\ne B^*$, which in this case is
equivalent to assuming that $B$ has deficiency indices $(1,1)$
\cite[chapter 4, section 1.2]{akhiezer1}. A classical result tells us
that the orthogonal polynomials of the first kind $P_k(z)$ associated
with (\ref{eq:jm}) are such that
\begin{equation*}
  \sum_{k=0}^\infty\abs{P_k(z)}^2<\infty
\end{equation*}
uniformly in any compact domain of the complex plane
\cite[theorem 1.3.2]{akhiezer1} .  Therefore, for any $z\in\mathbb{C}$,
$\pi(z):=\sum_{k=1}^\infty P_{k-1}(z)\delta_k$ is in $\cH$, and more
specifically in $\ker(B^*-zI)$ \cite[chapter 4, section 1.2]{akhiezer1}.
By construction of the polynomials of the first kind,
\begin{equation*}
  \inner{\pi(\cc{z})}{\delta_1}=P_0(\cc{z})\equiv 1\,,
\end{equation*}
so $B$ is a densely defined $0$-entire and $\delta_1$ is an entire gauge.

Now we outline how one may construct a $0$-entire operator which is
not densely defined. Let $B_0$ be the restriction of $B$ to the set
$\{\phi\in\dom(B):\inner{\phi}{\delta_1}=0\}$. It follows from
(\ref{eq:adjoint-def}), (\ref{eq:adjoint-shifted}) and
(\ref{eq:ker-adjoint}) that $\eta\in\ker(B_0^*-zI)$ if and only if it
satisfies the equation
\begin{equation*}
  \inner{B\phi}{\eta}=\inner{\phi}{z\eta}\qquad\forall\phi\in\dom(B_0)\,.
\end{equation*}
Thus $\ker(B_0^*-zI)$ is the set of $\eta$'s in $\cH$ that satisfy
\begin{equation}
  \label{eq:difference}
  b_{k-1}\inner{\delta_{k-1}}{\eta}+q_k\inner{\delta_{k}}{\eta} +
  b_k\inner{\delta_{k+1}}{\eta}=z\inner{\delta_{k}}{\eta}\quad\forall k>1
\end{equation}
Hence $\dim\ker(B_0^*-zI)\le 2$. Now, let
$\theta(z):=\sum_{k=1}^\infty Q_{k-1}(z)\delta_k$, where $Q_k(z)$ is
the $k$-th polynomial of second kind associated to (\ref{eq:jm}). By
the definition of the polynomials $P_k(z)$ and $Q_k(z)$
\cite[chapter 1, section 2.1]{akhiezer1}, $\pi(z)$ and $\theta(z)$ are linearly
independent solutions of (\ref{eq:difference}) for every fix
$z\in\C$. Moreover, since $B\ne B^*$, $\pi(z)$ and $\theta(z)$ are in
$\cH$ for all $z\in\C$ \cite[theorems 1.3.1 and 1.3.2]{akhiezer1},
\cite[theorem 3]{simon}. So one arrives at the conclusion that, for
every fix $z\in\C$,
\begin{equation*}
  \ker(B_0^*-zI)=\Span\{\pi(z),\theta(z)\}\,.
\end{equation*}
Any symmetric nonselfadjoint extension of $B_0$ has deficiency indices
(1,1) Furthermore, if $\kappa(z)$ is a ($z$-dependent) linear combination of
$\pi(z)$ and $\theta(z)$ such that
$\inner{\kappa(z)}{\theta(z)}=0$ for all $z\in\C\setminus\R$, then (by a
parametrized version of \cite[theorem 2.4]{simon}) there
corresponds to an appropriately chosen isometry from
$\Span\{\kappa(z)\}$ onto $\Span\{\kappa(\cc{z})\}$ a nonselfadjoint
symmetric extension $\widetilde{B}$ of $B_0$ such that
$\dom(\widetilde{B})$ is not dense and
$\ker(\widetilde{B}^*-zI)=\Span\{\theta(z)\}$.  We claim that $\widetilde{B}$
is a nondensely defined $0$-entire operator. Indeed,
$\widetilde{B}\in\ournewclass$ (the simplicity follows from the
properties of the associated polynomials \cite[chapter 1, addenda and
problem 7]{akhiezer1}). Moreover, since
\begin{equation*}
  \inner{\theta(\cc{z})}{\delta_2}=b_1^{-1}\,,\qquad\forall z\in\C\,,
\end{equation*}
$\delta_2$ is an entire gauge.
\end{example}

%%%%%%%%%%%%%%%%%%%%%%%%%%%%%

\subsection{On de Branges spaces with zero-free associated functions}
\label{subsec:dB}

Let $\cB$ denote a nontrivial Hilbert space of entire functions with
inner product $\inner{\cdot}{\cdot}_\cB$. $\cB$ is said to be a de
Branges space when, for every function $f(z)$ in $\cB$, the following
conditions holds:
\begin{enumerate}[({A}1)]
\item For every $w\in\C\setminus\R$, the linear functional
        $f(\cdot)\mapsto f(w)$  is continuous;

\item for every non-real zero $w$ of $f(z)$, the function
        $f(z)(z-\cc{w})(z-w)^{-1}$ belongs to $\cB$
        and has the same norm as $f(z)$;

\item the function $f^\#(z):=\cc{f(\cc{z})}$ also belongs to $\cB$
        and has the same norm as $f(z)$.
\end{enumerate}

In view of the Riesz lemma, (A1) is equivalent to the existence of a
reproducing kernel $k(z,w)$ that belongs to $\cB$ for every non-real
$w$ such that $\inner{k(\cdot,w)}{f(\cdot)}_{\cB} = f(w)$ for every
$f(z)\in\cB$. Also, for any $w\in\C$,
$k(w,w)=\inner{k(\cdot,w)}{k(\cdot,w)}_\cB\ge 0$ where, as a
consequence of (A2), the positivity is strict for every non-real $w$
unless $\cB\cong\mathbb{C}$; see the proof of \cite[theorem
23]{debranges}. Note that $k(z,w)=\inner{k(\cdot,z)}{k(\cdot,w)}_\cB$
whenever $z$ and $w$ are both non-real, therefore
$k(w,z)=\cc{k(z,w)}$. Furthermore, due to (A3) it can be proven (again
using \cite[theorem 23]{debranges}) that
$\cc{k(\cc{z},w)}=k(z,\cc{w})$ for every non-real $w$.  Also note that
$k(z,w)$ is entire with respect to its first argument and, by (A3), it
is anti-entire with respect to the second one (once $k(z,w)$, as a
function of its second argument, has been extended to the whole
complex plane \cite[problem 52]{debranges}).

There is an alternative definition of a de Branges space. Its starting
point is an entire function $e(z)$ of the Hermite-Biehler class, that
is, an entire function without zeros in the upper half-plane $\C^+$
that satisfies the inequality $\abs{e(z)}>\abs{e^\#(z)}$ for
$z\in\C^+$. On the basis of this function, one firstly defines $\cB(e)$ to be
the linear manifold of all entire functions $f(z)$ such that both
$f(z)/e(z)$ and $f^\#(z)/e(z)$ belong to the Hardy space $H^2(\C^+)$,
and secondly, equips it with the inner product
\[
\inner{f(\cdot)}{g(\cdot)}_{\cB(e)}:=
\int_{-\infty}^\infty\frac{\cc{f(x)}g(x)}{\abs{e(x)}^2}dx.
\]
Then $\cB(e)$ turns out to be a Hilbert space of entire functions.

Now, according to \cite[chapter 2]{debranges}, every space $\cB(e)$
obeys (A1--A3) and conversely, given a space $\cB$, there exists an
Hermite-Biehler function $e(z)$ such that $\cB$ coincides with
$\cB(e)$ as sets and the respective norms satisfy the equality
$\norm{f(\cdot)}_{\cB}=\norm{f(\cdot)}_{\cB(e)}$. Thus, both
definitions of de Branges spaces are equivalent.

\begin{remark}
\label{rem:alternative-def-dB}
  For an entire function $f(z)$, the condition that $f(z)/e(z)$ and
  $f^\#(z)/e(z)$ are in $H^2(\C^+)$ is equivalent to
  \begin{equation*}
    \int_{-\infty}^{\infty}\abs{\frac{f(x)}{e(x)}}^2dx<\infty
  \end{equation*}
and the functions $f(z)/e(z)$ and $f^\#(z)/e(z)$ being of bounded type
and nonpositive mean type in the upper half-plane
\cite[proposition 2.1]{remling}.
\end{remark}

The function $e(z)$ is not uniquely determined by the de Branges space
$\cB$. However, if one sets
\begin{equation*}
%\label{eq:e-given-by-k}
e(z)=-i\sqrt{\frac{\pi}{k(w_0,w_0)\im(w_0)}}
\left(z-\cc{w_0}\right)k(z,w_0),
\end{equation*}
where $w_0$ is some fixed complex number in $\C^+$, then $\cB=\cB(e)$
in the sense given above.

An entire function $g(z)$ is said to be associated to a de Branges
space $\cB$ if for all $f(z)\in\cB$ and
$w\in\C$,
\begin{equation*}
\frac{g(z)f(w)-g(w)f(z)}{z-w}\in\mathcal{B}.
\end{equation*}
The set of associated functions is denoted $\assoc\mathcal{B}$.  It
can be shown that
\begin{equation}
\label{eq:associated-functions}
\assoc\mathcal{B} = \mathcal{B} + z\mathcal{B};
\end{equation}
see \cite[theorem 25]{debranges} and \cite[lemma 4.5]{kaltenback} for
alternative characterizations. Incidentally, let us note that
$e(z)\in\assoc\mathcal{B}(e)\setminus\mathcal{B}(e)$; this fact
follows straightforwardly from \cite[theorem 25]{debranges}.

The space $\assoc\mathcal{B}(e)$ contains a distinctive family of
entire functions:
\begin{equation*}
%\label{eq:functions-s}
s_\beta(z):=\frac{i}{2}\left[e^{i\beta}e(z)-e^{-i\beta}e^\#(z)\right],
%	= -a(z)\sin\beta + b(z)\cos\beta
	\quad \beta\in{[}0,\pi).
\end{equation*}
These real entire functions are related to the selfadjoint extensions
of the multiplication operator $S$ defined
by \begin{equation}\label{eq:multiplication-operator}
  \dom(S):=\{f(z)\in\mathcal{B}: zf(z)\in\mathcal{B}\},\quad
  (Sf)(z)=zf(z).
\end{equation}
The operator $S$ is closed, symmetric with deficiency indices $(1,1)$,
and its domain is not necessarily dense in $\mathcal{B}$
\cite[proposition 4.2]{kaltenback}. Furthermore, $S$ is regular
\cite[corollary 4.7]{kaltenback} and hence simple. It
turns out that $\cc{\dom(S)}\neq\mathcal{B}$ if and only if there
exists $\gamma\in{[}0,\pi)$ such that
$s_\gamma(z)\in\mathcal{B}$. Moreover,
$\dom(S)^\perp=\Span\{s_\gamma(z)\}$ \cite[theorem 29]{debranges} and
\cite[corollary 6.3]{kaltenback}; compare with (i) of
proposition~\ref{prop:misc-about-symm-operators}.

Given a selfadjoint extension $S_\sharp$ of $S$, one can find a
unique $\beta$ in $[0,\pi)$ such that
\begin{equation}
\label{eq:selfadjoint-extensions-s}
(S_\sharp-wI)^{-1}f(z)
	= \frac{f(z)-\frac{s_\beta(z)}{s_\beta(w)}f(w)}{z-w},\quad
		w\not\in\Sp(S_\sharp),\quad f(z)\in\cB.
\end{equation}
with $\Sp(S_\sharp)=\left\{x\in\mathbb{R}: s_\beta(x)=0\right\}$
\cite[propositions 4.6 and 6.1]{kaltenback}.  When $S_\sharp$ is a
selfadjoint operator extension of $S$, then
\eqref{eq:selfadjoint-extensions-s} is equivalent to
\[
\dom(S_\sharp) =
	\left\{g(z)=\frac{f(z)-
\frac{s_\beta(z)}{s_\beta(z_0)}f(z_0)}{z-z_0},
	\quad f(z)\in\mathcal{B},\quad z_0:s_\beta(z_0)\neq 0\right\},
%	\label{eq:non-standard}
\]
and
\[
(S_\sharp g)(z) = z g(z) +
\frac{s_\beta(z)}{s_\beta(z_0)}f(z_0).\nonumber
\]
In this context, the function
\[
g_n(z):=\frac{s_\beta(z)}{z-x_n}.
\]
is the eigenfunction of $S_\sharp$ corresponding to $x_n\in\Sp(S_\sharp)$.
%(see \cite{kaltenback} for the proofs of these results and alternative
%proofs based on von Neumann classical theory in \cite[Proposition
%3.8]{II} for the case $\cc{\dom(S)}=\cB$ ).
Hence, due to the fact that $S$ is regular and simple, every
$s_\beta(z)$ has only real zeros of multiplicity one and the
zeros of any pair $s_\beta(z)$ and $s_{\beta'}(z)$ always
interlace.

The classical notion of associated functions \eqref{eq:associated-functions}
has been generalized in \cite{langer-woracek} as follows.
For $n\in\Z^+$ let
\begin{equation}
\label{eq:n-assoc-functions}
\assoc_n\mathcal{B} := \mathcal{B} + z\mathcal{B}+\cdots+z^n\mathcal{B}.
\end{equation}
These linear sets of so-called $n$-associated functions were
introduced in the context of intermediate Weyl coefficients and have
been thoroughly studied in \cite{langer-woracek, woracek2}. Moreover,
for any $n\in\Z^+$, one has necessary and sufficient conditions for
the existence of a real zero-free entire functions in the space
$\assoc_n\mathcal{B}$. The statement of this important result (see
theorem~\ref{thm:1-in-assoc-n-boosted} below) is essentially
theorem~3.2 of \cite{woracek2} with a slight modification justified by
lemmas 3.3 and 3.4 of \cite{II}. See also \cite{woracek} for a more
elementary (and restricted) version of this theorem.

\begin{theorem}
\label{thm:1-in-assoc-n-boosted}
  Suppose $e(x)\neq 0$ for $x\in\mathbb{R}$ and
  $e(0)=(\sin\gamma)^{-1}$ for some fixed $\gamma\in(0,\pi)$. Let
  $\{x_j\}_{j\in\mathbb{N}}$ be the sequence of zeros of the function
  $s_\gamma(z)$. Also, let $\{x_j^+\}_{n\in\mathbb{N}}$ and
  $\{x_j^-\}_{n\in\mathbb{N}}$ be the sequences of positive,
  respectively negative, zeros of $s_\gamma(z)$, arranged according
  to increasing modulus.  Then a zero-free, real entire function
  belongs to $\assoc_n\mathcal{B}(e)$ if and only if the following
  conditions hold true:
\begin{enumerate}[(C1)]
\item The limit
	$\displaystyle{\lim_{r\to\infty}\sum_{0<|x_j|\le r}
		\frac{1}{x_j}}$
	exists.
\item $\displaystyle{\lim_{j\to\infty}\frac{j}{x_j^{+}}
		=- \lim_{j\to\infty}\frac{j}{x_j^{-}}<\infty}$.
\item Assuming that $\{b_j\}_{n\in\mathbb{N}}$ are the zeros of
  $s_\beta(z)$, define
	\[
	h_\beta(z):=\left\{\begin{array}{ll}
			\displaystyle{\lim_{r\to\infty}\prod_{|b_j|\le r}
			\left(1-\frac{z}{b_j}\right)}
				& \mbox{ if 0 is not a root of } s_\beta(z),
			\\
			\displaystyle{z\lim_{r\to\infty}\prod_{0<|b_j|\le r}
			\left(1-\frac{z}{b_j}\right)}
				& \mbox{ otherwise. }
			   \end{array}\right.
	\]
	The series
	$\displaystyle{
		\sum_{j\in\mathbb{N}}\abs{\frac{1}
		{x_j^{2n}h_{0}(x_j)h_{\gamma}'(x_j)}}}$ is convergent.
\end{enumerate}
\end{theorem}

\begin{remark}
\label{rem:zero-free-implies-real-zero-free}
By a simple argument due to H. Woracek (private communication), if
there is a zero-free function in a de Branges space $\mathcal{B}(e)$,
then there is a real zero-free function in $\mathcal{B}(e)$. Indeed, let
$f(z)\in\mathcal{B}(e)$ be zero-free. Then $f(z)/f^\#(z)$ is an
entire, zero-free function of bounded type in the upper half-plane (see
remark~\ref{rem:alternative-def-dB}). By
\cite[theorem 6.17]{RR}, one has
\begin{enumerate}[(a)]
  \item $f(z)/f^\#(z)$ is of exponential type,
  \item $\displaystyle\int_{-\infty}^\infty
\frac{\log^+\abs{f(x)/f^\#(x)}}{1+x^2}dx<\infty$.
\end{enumerate}
In view of (a), the Hadamard factorization theorem yields
$f(z)/f^\#(z)=Ce^{(a+\I b)z}$ with $a,b\in\mathbb{R}$, but (b) implies
that $a=0$. Thus $f(z)/f^\#(z)=Ce^{\I bz}$. Clearly, it suffices to
consider the case $b>0$ since if $b=0$, then $f(z)$ is real; and if
$b<0$, then one considers $f^\#(z)/f(z)$ instead of
$f(z)/f^\#(z)$. Now the entire function $
g(z):=f(z)e^{-\I\frac{a}{2}z}$ is real, zero-free and it is
straightforward to verify that
\begin{equation*}
\int_{-\infty}^\infty\abs{\frac{g(x)}{e(x)}}^2dx<\infty
\end{equation*}
and the quotients $g(z)/e(z)$, $g(z)^\#/e(z)$ are of bounded type and
nonpositive mean type in the upper half-plane. According to
remark~\ref{rem:alternative-def-dB}, this means that $g(z)\in\cB(e)$.
\end{remark}

\begin{remark}
\label{rem:assoc-dB-space}
As discussed in \cite{langer-woracek}, every one of the linear set
$\assoc_n\mathcal{B}(e)$ can be turned into a de Branges space.
In fact, from corollary 3.4 of \cite{langer-woracek} it follows that
\[
\assoc_n\mathcal{B}(e(z)) = \mathcal{B}((z+w)^ne(z)),
\]
as sets, for any $w\in\C^+$. This fact will be used later in
Section \ref{sec:gelfand-triplets}.
\end{remark}

\begin{remark}
\label{rem:charac-zero-free}
The two previous remarks can be used to sharpen
theorem~\ref{thm:1-in-assoc-n-boosted}. Namely, if
$\assoc_n\mathcal{B}(e)$ contains a (possibly non-real) zero-free
function, then conditions (C1), (C2), and (C3) are fulfilled.
\end{remark}

%%%%%%%%%%%%%%%%%%%%%%%%%%%%%
\subsection{\texorpdfstring{A functional model for operators in $\ournewclass$}
		{A functional model}}
\label{subsec:functional-model}

In this subsection we construct a functional model following the
framework developed in \cite{II}, but having adapted it to comprise
all the operators in the class $\ournewclass$. This functional model
is based on Krein's representation theory \cite[theorems 2 and
3]{krein1}, \cite[section 1.2]{gorbachuk}, but differs from it in a
crucial way as commented in
remark~\ref{rem:funct-model-not-Kreins}. It is worth mentioning that
there is an alternative functional model for the same class
$\ournewclass$ recently developed in \cite{martin}.
Some of the material in this subsection can also be found in
\cite{III}.

The functional model described below rests on the properties of the
generalized Cayley transform given in
Proposition~\ref{prop:misc-about-symm-operator2} with the following
addition.

\begin{proposition}
%\label{prop:on-psi-w}
  Let $A$ be an element of $\ournewclass$ and $J$ be an involution
  that commutes with one of the canonical selfadjoint extensions of $A$
  (hence with all of them), say, $A_\gamma$. For every $v\in\Sp(A_\gamma)$,
  there exists $\psi_v\in\Ker(A^*-vI)$ such that
  $J\psi_v=\psi_v$.
\end{proposition}
\begin{proof}
  Let $\phi_v$ be a nontrivial element of $\Ker(A_\gamma-vI)$. It
  follows from the fact that $J$ commutes with $A_\gamma$ that
  $J\phi_v\in\Ker(A_\gamma-vI)$. But, by our assumption on the
  deficiency indices of $A$ and its regularity, the subspace
  $\Ker(A^*-vI)$ is one-dimensional and it contains
  $\Ker(A_\gamma-vI)$. So $J$, restricted to $\ker(A^*-vI)$, reduces
  to multiplication by a scalar $\alpha$ and the properties of the
  involution imply that $\abs{\alpha}=1$. Now, $\psi_v:=
  (1+\alpha)\phi_v$ has the required properties.
\end{proof}

For any $A\in\ournewclass$ and a fixed involution $J$ that commutes
with the selfadjoint extensions of $A$ within $\cH$, define
\begin{equation}
\label{eq:xi-def}
\xi_{\gamma,v}(z)
	:=h_\gamma(z)\left[I+(z-v)(A_\gamma-zI)^{-1}\right]\psi_v\,,
\end{equation}
where $v$ and $\psi_v$ are chosen as in the previous proposition, and
$h_\gamma(z)$ is a real entire function whose zero set is
$\Sp(A_\gamma)$ (see part (i) of
proposition \ref{prop:properties-of-new-class}). Clearly, up to a
zero-free real entire function,
$\xi_{\gamma,v}(z)$ is completely determined by the choice of the
selfadjoint extension $A_\gamma$ and $v$. In fact, as it is stated
more precisely below, $\xi_{\gamma,v}(z)$ does not depend on
$A_\gamma$ nor on $v$.
\begin{proposition}
\label{prop:xi-properties}
For the function defined in (\ref{eq:xi-def}), the following holds:
\begin{enumerate}[(i)]
\item The vector-valued function $\xi_{\gamma,v}(z)$ is zero-free and
	entire. It lies in $\Ker(A^* - zI)$ for all $z\in\C$.
\item $J\xi_{\gamma,v}(z)=\xi_{\gamma,v}(\cc{z})$ for every $z\in\C$.
\item Given $\xi_{\gamma_1,v_1}(z)$ and
  $\xi_{\gamma_2,v_2}(z)$, there exists a zero-free real entire
  function $g(z)$ such that $\xi_{\gamma_2,v_2}(z)=
  g(z)\xi_{\gamma_1,v_1}(z)$.
\end{enumerate}
\end{proposition}
\begin{proof}
  In view of Proposition~\ref{prop:misc-about-symm-operator2}, the proof
  of (i) is rather straightforward. In fact, one should only follow
  the first part of the proof of \cite[lemma 4.1]{II}. The proof of
  (ii) also follows easily from our choice of $\psi_v$ and
  $h_\gamma(z)$ in the definition of $\xi_{\gamma,v}(z)$. To prove
  (iii), one first uses
  proposition~\ref{prop:misc-about-symm-operator2} and the fact that
  $\dim\Ker(A^*-vI)=1$ to obtain that $\xi_{\gamma_2,v_2}(z)$ and
  $\xi_{\gamma_1,v_1}(z)$ differ by a nonzero scalar complex
  function. Then the reality of this function follows from (ii).
\end{proof}

Due to (iii) of proposition \ref{prop:xi-properties}, from now on
the function $\xi_{\gamma,v}(z)$ will be denoted by
$\xi(z)$. Actually, the proof of (iii) leads to the following remark.

\begin{remark}
\label{rem:uniqueness-xi}
Every vector-valued entire function satisfying (i) and (ii) is unique
up to a zero-free real entire function. Moreover, if a vector-valued
entire function satisfies (i), then, for the involution constructed in
proposition~\ref{prop:existence-of-commuting-involution}, it also
complies with (ii).
\end{remark}

On the basis of the function $\xi(z)$ that we have constructed, let us
now define
\begin{equation*}%\label{eq:defining-phi}
\left(\Phi\varphi\right)(z):=\inner{\xi(\cc{z})}{\varphi},\qquad
	\varphi\in\cH.
\end{equation*}
$\Phi$ maps $\cH$ onto a certain linear manifold $\widehat{\cH}$ of
entire functions. Since $A$ is simple, it follows that $\Phi$ is
injective.  A generic element of $\widehat{\cH}$ will be denoted
by $\widehat{\varphi}(z)$, as a reminder of the fact that it is the
image under $\Phi$ of a unique element $\varphi\in\cH$. Clearly,
the linear space $\widehat{\cH}$ is turned into a Hilbert space by defining
\begin{equation*}
  \inner{\widehat{\eta}(\cdot)}{\widehat{\varphi}(\cdot)}:=
\inner{\eta}{\varphi}\,,
\end{equation*}
and $\Phi$ is an isometry from $\cH$ onto $\widehat{\cH}$.
\begin{proposition}
%\label{prop:de-branges-space}
$\hat{\cH}$ is a de Branges space.
\end{proposition}
\begin{proof}
It suffices to show that the axioms given at the beginning of
Section \ref{subsec:dB} holds for $\hat{\cH}$.

It is straightforward to verify that $k(z,w):=\inner{\xi(\cc{z})}{\xi(\cc{w})}$
is a reproducing kernel for $\hat{\cH}$. This accounts for (A1).

Suppose $\hat{\varphi}(z)\in\hat{\cH}$ has a zero at $z=w$. Then its preimage
$\varphi\in\cH$ lies in $\ran(A-wI)$. This allows one
to set $\eta\in\cH$ by
\[
\eta  = (A-\cc{w}I)(A-wI)^{-1}\varphi
	= \varphi + (w-\cc{w})(A_\gamma-wI)^{-1}\varphi.
\]
Now, recalling \eqref{eq:xi-def} and applying the resolvent identity, one
obtains
\[
\inner{\xi(\cc{z})}{\eta} = \frac{z-\cc{w}}{z-w}\inner{\xi(\cc{z})}{\varphi}.
\]
Since $\eta$ and $\varphi$ are related by a Cayley transform, the equality
of norms follows. This proves (A2).

As for (A3), consider any $\hat{\varphi}(z)=\inner{\xi(\cc{z})}{\varphi}$.
Then, as a consequence of (ii) of proposition~\ref{prop:xi-properties},
one has $\hat{\varphi}^\#(z)=\inner{\xi(\cc{z})}{J\varphi}$.
\end{proof}
\begin{remark}
\label{rem:image-of-involution}
The last part of the proof given above
shows that $^\#=\Phi J \Phi^{-1}$.
\end{remark}
The following statement is obvious, but it gives the indispensable
properties of any functional model so we bring it up here for the sake
of completeness.
\begin{proposition}
%  \label{prop:operator-multiplication}
If $S$ is the multiplication operator in
  $\widehat{\cH}$ given by (\ref{eq:multiplication-operator}), then
\begin{enumerate}[{(i)}]
\item $S=\Phi A\Phi^{-1}$ and $\dom(S)=\Phi\dom(A)$.
\item The selfadjoint extensions of $S$ within $\hat{\cH}$ are in one-one
	correspondence with the selfadjoint extensions of $A$ within $\cH$.
\end{enumerate}
\end{proposition}
\begin{remark}
\label{rem:funct-model-not-Kreins}
The functional model we have constructed yields a de Branges
space for every operator in $\ournewclass$. In contrast, Krein's
representation theory yields a de Branges space only when the operator
is in $\nentireclass{0}$.
\end{remark}
In the previous subsection we explained that the operator of
multiplication $S$ in a de Branges space $\mathcal{B}$ is in
$\mathscr{S}(\mathcal{B})$. Now, the functional model we have
constructed tells us that every element in $\ournewclass$ is unitarily
equivalent to the multiplication operator in a certain de Branges
space. Although this assertion is also present in \cite{martin}, our
functional model is simpler and more straightforward.

%%%%%%%%%%%%%%%%%%%%%%%%%%%%%%%%%%%%%%%%%%%%%%%%%%%%%%%%%%%%%%%%%%%%%%%%%%%%%%%

\section{\texorpdfstring{Characterization of $n$-entire operators}
		{Characterization of n-entire operators}}
\label{sec:main-results}
This section provides various sets of necessary and sufficient
conditions for an operator in $\ournewclass$ to be in
$\nentireclass{n}$. We heavily rely on the functional model we have
constructed above for our characterizations.
\begin{proposition}
\label{prop:when-the-operator-is-entire}
$A\in\ournewclass$ is $n$-entire if and only if $\assoc_n\widehat{\cH}$
contains a zero-free entire function.
\end{proposition}
\begin{proof}
Let $m(z)\in\assoc_n\widehat{\cH}$ be the function whose existence is assumed.
Such function can be written as $m(z)=m_0(z)+zm_1(z)+\cdots+z^nm_n(z)$
for some functions $m_0(z),m_1(z),\ldots,m_n(z)\in\widehat{\cH}$, each of
them in turn satisfying $m_j(z)=\inner{\xi(\cc{z})}{\mu_j}$ for some
$\mu_j\in\cH$. Therefore, $\mu_0+z\mu_1+\cdots+z^n\mu_n$ is
never orthogonal to $\Ker(A^*-zI)$ for all $z\in\mathbb{C}$.

The proof of the necessity is rather obvious hence omitted.
\end{proof}

\begin{remark}
\label{rem:real-gauge}
Krein asserted without proof that if a densely defined
operator is in $\nentireclass{0}$, then one can always find a gauge
$\mu$ that commutes with the involution $J$ of
proposition~\ref{prop:existence-of-commuting-involution} ($\mu$ is a
{\em real} entire gauge) \cite[theorem 8]{krein2}. The proof actually
follows directly from our construction by means of
remarks~\ref{rem:zero-free-implies-real-zero-free} and
\ref{rem:image-of-involution} since the image under $\Phi$ of an entire
gauge is a zero-free function.
\end{remark}
\begin{example}
\label{example:linear-momentum}
In $\cH=L^2[-a,a]$, $0<a<+\infty$, consider the operator
\begin{equation*}
\dom(A)=\{\varphi(x)\in\text{AC}[-a,a]:\varphi(a)=0=\varphi(-a)\},
\quad A:=i\frac{d}{dx}.
\end{equation*}
Clearly, $A$ is closed and symmetric. Moreover,
\begin{equation*}
\dom(A^*) = \text{AC}[-a,a],\quad A^*=i\frac{d}{dx},
\end{equation*}
from which it is straightforward to verify that the deficiency indices
of $A$ are $(1,1)$.  The canonical selfadjoint extensions of $A$ can
be parametrized as
\begin{equation*}
\dom(A_\gamma) = \{\varphi(x)\in\text{AC}[-a,a]:
			\varphi(a)=e^{-i2\gamma}\varphi(-a)\},
\quad A_\gamma = i\frac{d}{dx},
\end{equation*}
for $\gamma\in{[}0,\pi)$. These selfadjoint extensions correspond to
different realizations of the linear momentum operator within the
interval $[-a,a]$. By a straightforward calculation,
\begin{equation}
\label{eq:spectrum-linear-momentum}
\Sp(A_\gamma)=\left\{\frac{\gamma+k\pi}{a}:k\in\Z\right\}.
\end{equation}
Clearly, the spectra are interlaced and their union equals $\R$ so it
follows that $A$ is regular, hence simple (see subsection 2.1).

Let us define $\xi(x,z):=e^{-izx}$, $x\in [-a,a]$, $z\in\C$. This
zero-free entire function belongs to $\Ker(A^* - zI)$ for all
$z\in\C$. By remark \ref{rem:uniqueness-xi} and
proposition~\ref{prop:when-the-operator-is-entire}, for proving that $A$
is 1-entire, it suffices to find $\mu_0(x),\mu_1(x)\in L^2[-a,a]$
such that
\begin{equation}
\label{eq:nice-example}
\int_{-a}^{a}e^{-iyx}\mu_0(x)dx + y \int_{-a}^{a}e^{-iyx}\mu_1(x)dx =1
\end{equation}
for all $y\in\R$ (and then use analytic continuation to the whole
complex plane). Our searching will be guided by formally taking the
inverse Fourier transform of \eqref{eq:nice-example} and switching
without much questioning the order of integration, obtaining in that
way the differential equation
\[
\mu_0(x) - i\mu'_1(x) = \delta(x).
\]
This equation suggests to set
\begin{eqnarray}
\mu_0(x) &= \frac{1}{2a}\chi_{[-a,a]}(x)\label{eq:mu-0}
\\[1mm]
\mu_1(x) &= - i\frac{a+x}{2a}\chi_{[-a,0]}(x)
		    + i\frac{a-x}{2a}\chi_{[0,a]}(x),\label{eq:mu-1}
\end{eqnarray}
where $\chi_S(x)$ denotes the characteristic function of the set $S$.
A simple computation shows that indeed \eqref{eq:mu-0} and
\eqref{eq:mu-1} satisfy \eqref{eq:nice-example}. Thus, it has been
proven that $A\in\nentireclass{1}$, and below, in
example~\ref{example:linear-momentum2} it will be shown that
$A\not\in\nentireclass{0}$.
\end{example}

\begin{example}
\label{example:laplacian}
In $\cH=L^2[0,a]$, $0<a<+\infty$, we consider the operator
\[
D:= -\frac{d^2}{dx^2},
\]
with domain
\[
\dom(D) = \left\{\varphi(x)\in\text{AC}^2[0,a]:\varphi'(0)=0,
			\varphi(a)=\varphi'(a)=0\right\}.
\]
This operator is symmetric and has deficiency indices $(1,1)$. The
adjoint operator $D^*$ is given by the same differential expression as
$D$ but with domain
\[
\dom(D^*) = \{\varphi(x)\in\text{AC}^2[0,a]:\varphi'(0)=0\}.
\]
The selfadjoint restriction of $D^*$ can be parametrized by
$\beta\in[0,\pi)$ and are given by
\[
%\label{eq:selfadjoint-laplacian}
D_\beta:= -\frac{d^2}{dx^2},
\]
with domain
\[
%\label{eq:selfadjoint-laplacian-domain}
\dom(D_\beta) = \left\{\varphi(x)\in\text{AC}^2[0,a]:\varphi'(0)=0,
			\varphi(a)\sin\beta+\varphi'(a)\cos\beta=0\right\}.
\]
That is, the operators $D_\beta$ are the (selfadjoint) realizations of
the Laplacian operator in the interval $[0,a]$ with Neumann boundary
condition at $x=0$. The spectra of these operators are simple and
discrete. Moreover, they are pairwise interlaced, so $A$ is
regular and therefore simple.

The function $\xi(x,z):=\cos(\sqrt{z}x)$ is the (unique) solution of
the equation
\[
-\xi''(x,z) = z \xi(x,z),\quad z\in\C,
\]
with boundary conditions $\xi(0,z)=1$ and $\xi'(0,z)=0$. Hence this
entire function belongs to $\Ker(D^*-zI)$ for every $z\in\C$.  We
will show that there exist functions $\mu_0(x),\mu_1(x)\in
L^2[0,a]$ such that
\begin{equation}
\label{eq:trick-revisited}
\int_0^a\cos(yx)\mu_0(x)dx
+ y^2\int_0^a\cos(yx)\mu_1(x)dx = 1,\quad y\in\R^+.
\end{equation}
By identifying $y=\sqrt{z}$, and then by analytic continuation from
$z\in\R^+$ to $\C$, this will prove that $D$ is 1-entire. To find the
functions $\mu_0(x),\mu_1(x)$, we use the same heuristic approach of
the previous example.

We assume that $\mu_0(x)$ and $\mu_1(x)$ are even functions on the interval
$[-a,a]$. Then \eqref{eq:trick-revisited} is equivalent to
\[
\frac12\int_{-a}^a e^{-iyx}\mu_0(x)dx + y^2\frac12\int_{-a}^a
e^{-iyx}\mu_1(x)dx = 1,
\]
where now this equation can be considered valid for all $y\in\R$. Then
we take the Fourier transform to obtain the formal differential equation
\[
\mu_0(x) - \mu_1''(x) = 2\delta(x),
\]
a solution of which is given by (the even extension of)
\[
\mu_0(x)=\frac{1}{a}\chi_{[0,a]}(x),\quad \mu_1(x)=\frac{1}{2a}(x-a)^2\chi_{[0,a]}(x).
\]
A straightforward computation shows that these functions indeed
fulfill \eqref{eq:trick-revisited}.
\end{example}
The following may be considered as an alternative definition
of a densely defined $n$-entire operator.

\begin{proposition}
\label{prop:another-characterization}
A densely defined operator $A$ is in $\nentireclass{n}$ if and only if
there exists a collection $\mu_0,\dots,\mu_n\in\cH$ such that
\begin{equation}
\label{eq:criterion2}
  \sum_{j=0}^n\inner{(A^*)^j\xi(\cc{z})}{\mu_j}\ne 0
\end{equation}
for all $z\in\C$.
\end{proposition}
\begin{proof}
Since $A^*$ is an operator, part (i) of proposition~\ref{prop:xi-properties}
becomes $A^*\xi(z)=z\xi(z)$ hence $(A^*)^j\xi(z)=z^j\xi(z)$, $j\in\Z^+$.
Given $\mu_0,\dots,\mu_n\in\cH$, one has the identity
\[
\inner{\xi(\cc{z})}{\mu_0+z\mu_1+\cdots+z^n\mu_n}
	= \sum_{j=0}^n\inner{\cc{z}^j\xi(\cc{z})}{\mu_j}
	= \sum_{j=0}^n\inner{(A^*)^j\xi(\cc{z})}{\mu_j}.
\]
The statement then follows.
\end{proof}

\begin{remark}
The previous proposition can be extended to operators with non-dense domain
provided that \eqref{eq:criterion2} is written in terms of
the operator part of the adjoint relation. See \cite{arens} for more details.
\end{remark}

\begin{proposition}
\label{prop:spectrum-tells-if-operator-is-n-entire}
For $A\in\ournewclass$, consider the selfadjoint extensions (within
$\cH$) $A_0$ and $A_{\gamma}$, with $0<\gamma<\pi$. Then $A$ is
$n$-entire if and only if $\Sp(A_0)$ and $\Sp(A_{\gamma})$ obey
conditions (C1), (C2) and (C3) of
theorem~\ref{thm:1-in-assoc-n-boosted}.
\end{proposition}
\begin{proof}
  Apply theorem~\ref{thm:1-in-assoc-n-boosted} and
  remark~\ref{rem:charac-zero-free} along with
  proposition~\ref{prop:when-the-operator-is-entire}.
\end{proof}

\begin{example}
\label{example:linear-momentum2}
Consider the operator $A$ and its selfadjoint extensions $A_\gamma$,
with $\gamma\in [0,\pi)$, given in example~\ref{example:linear-momentum}.
Taking into account (\ref{eq:spectrum-linear-momentum}), we now direct
our attention to conditions (C1), (C2) and (C3) of
proposition~\ref{prop:spectrum-tells-if-operator-is-n-entire}. (C1)
and (C2) are trivially fulfilled. As for (C3), we choose
$\gamma=\pi/2$ and notice that
\begin{equation*}
h_{\pi/2}(z) = \lim_{m\to\infty}\prod_{k=1}^m
			\left(1-\frac{4a^2z^2}{\pi^2(4k^2-4k+1)}\right)
	     = \cos(az),
\end{equation*}
while a similar computation shows that $h_0(z)=\sin(az)$. This implies
that (C3) is satisfied only for $n\ge 1$ . That is,
$A\in\nentireclass{1}\setminus\nentireclass{0}$.
\end{example}

\begin{example}
  Let us return to the Laplacian operator $D$ given in
  example~\ref{example:laplacian}. For any
  $\beta\in[0,\pi)$, an eigenvalue $b$ of the selfadjoint extension
  $D_\beta$ satisfies the identity
\[
\sqrt{b}\frac{\sin\sqrt{b}a}{\cos\sqrt{b}a} = \frac{\sin\beta}{\cos\beta}.
\]
In particular,
\[
\Sp(D_0) = \left\{\frac{\pi^2 k^2}{a^2}:k\in\N\right\},\quad
\Sp(D_{\pi/2}) =
\left\{\frac{\pi^2}{a^2}\left(\frac{2k-1}{2}\right)^2\!:k\in\N\right\}.
\]
Conditions (C1) and (C2) of
proposition~\ref{prop:spectrum-tells-if-operator-is-n-entire} are
clearly fulfilled. As for (C3), a computation shows that
\[
h_0(z) = \lim_{m\to\infty}\prod_{k=1}^m \left(1-\frac{a^2z}{\pi^2 k^2}\right)
	   = \frac{\sin a\sqrt{z}}{a\sqrt{z}},
\]
and similarly $h_{\pi/2}(z)=\cos a\sqrt{z}$. Hence, the series
that defines (C3) is convergent as long as
$n\ge 1$. That is, $D\in\nentireclass{1}\setminus\nentireclass{0}$.
\end{example}

The result of the last example can be extended to the canonical
selfadjoint extensions of the Schr\"odinger operator
\begin{equation}
\label{eq:schrodinger-neumann}
H := -\frac{d^2}{dx^2} + V(x)
\end{equation}
where $V(x)\in L^1[0,a]$, with domain
\begin{equation}
\label{eq:schrodinger-neumann-domain}
\dom(H) = \left\{\varphi(x)\in\text{AC}^2[0,a]:\varphi'(0)=0,
			\varphi(a)=\varphi'(a)=0\right\}.
\end{equation}
In a suitable sense, $V(x)$ is a small perturbation of the Laplacian operator.
Due to this fact, it is shown in theorem 4.1 of \cite{remling} that the
de Branges space associated to the operator $H$ is as a set equal to
the one associated to the Laplacian operator. As a consequence of this, one
can formulate the following assertion.

\begin{corollary}
  In $\cH=L^2[0,a]$, $a>0$, every Schr\"odinger operator given by
  \eqref{eq:schrodinger-neumann} with $V(x)\in L^1[0,a]$ and domain
  \eqref{eq:schrodinger-neumann-domain} belongs to
  $\nentireclass{1}\setminus\nentireclass{0}$.
\end{corollary}

With some additional little work, this result can further be extended
to all Schr\"odinger operators arising from regular differential
expressions. In connection with this, see theorem 10.7 of
\cite{remling}.

\begin{proposition}
\label{prop:relaxed-condition-nentireness}
Assume that, for $A\in\ournewclass$, one can find a collection
$\eta_0,\eta_1,\ldots,\eta_n\in\cH$ such that
\eqref{eq:n-entireness-algebraic} is fulfilled for all $z\in\C$ except
a finite set of points.  Then $A$ is $n$-entire.
\end{proposition}
\begin{proof}
In view of the functional model introduced above, noting that
\[
\Phi(\eta_0+z\eta_1+\cdots+z^n\eta_n)
	= \widehat{\eta}_0(z) + z\widehat{\eta}_1(z)
		+ \cdots + z^n\widehat{\eta}_n(z),
\]
and recalling \eqref{eq:n-assoc-functions}, it suffices to consider the
case of a de Branges space $\cB$ such that $\assoc_n\cB$ contains a
non-trivial entire function having a finite number of roots. Suppose such a
function $g(z)\in\assoc_n\cB$ exists. Let $z_1,z_1,\ldots,z_k$ be its zeros
whose respective (necessarily finite) multiplicities are $m_1,m_1,\ldots,m_k$.
Since $\assoc_n\cB$ is division invariant \cite[lemma 2.11]{woracek2},
one has
\[
f(z):=\frac{g(z)}{(z-z_1)^{m_1}(z-z_2)^{m_2}\cdots(z-z_k)^{m_k}}
	\in\assoc_n\cB
\]
and it is zero-free. This completes the proof.
\end{proof}

\begin{remark}
  For every $A\in\ournewclass$ and every $n\in\Z^+$ one can always
  find a set $\eta_0,\eta_1,\ldots,\eta_n\in\cH$ such that
  \eqref{eq:n-entireness-algebraic} is fulfilled for all $z\in\C$
  except a countable set of points. However, in view of
  proposition~\ref{prop:spectrum-tells-if-operator-is-n-entire}, it is
  clear that there are operators in $\ournewclass$ that are not
  $n$-entire (just consider an operator having a canonical selfadjoint
  extension whose spectrum does not satisfies one of the conditions
  (C1) or (C2) stated there; see example below). We therefore conclude
  that the statement of
  proposition~\ref{prop:relaxed-condition-nentireness} is sharp.
\end{remark}

\begin{example}
\label{example:counter-ex}
  The following is an example of an operator in $\ournewclass$ but not
  in $\nentireclass{n}$ for any $n\in\Z^+$. In this case
  $\cH=L^2(\R)$. Our starting point is the harmonic oscillator
  operator
\[
H_0 := -\frac{d^2}{dx^2} + x^2
\]
with its usual domain of selfadjointness $\dom(H_0)$.  Let
$\kappa(x)\in L^2(\R)$ be a cyclic vector for $H_0$, for example
\[
\kappa(x) = \sum_{n=0}^\infty \frac{1}{(n!)^{1/2}}\phi_n(x) =
\frac{1}{\pi^{1/4}}e^{-\frac12(x^2-2\sqrt{2}x+1)},
\]
where $\{\phi_n(x)\}_{n=0}^\infty$ is the basis of normalized
eigenvectors of $H_0$.  Consider the family of rank-one perturbations
of $H_0$,
\[
H_\beta := H_0 +
\beta\inner{\kappa(\cdot)}{\cdot}_{L^2(\R)}\kappa(x),\quad \beta\in\R.
\]
Clearly all these operators are selfadjoint with domain
$\dom(H_\beta)=\dom(H_0)$. Moreover, by \cite{hassi1}
these operators are all canonical selfadjoint extensions of
\[
H := -\frac{d^2}{dx^2} + x^2,\quad
	\dom(H) = \left\{\varphi(x)\in\dom(H_0):
			\inner{\kappa(\cdot)}{\varphi(\cdot)}_{L^2(\R)}=0\right\}.
\]
$H$ is a closed regular symmetric operator with deficiency indices
$(1,1)$ (the regularity follows from the cyclicity of
$\kappa(x)$). Observe that $\dom(H)$ is not dense in $L^2(\R)$. Thus,
there is an exceptional selfadjoint extension of $H$ that is not an
operator; this extension corresponds to the ``infinite coupling''
$\beta=\infty$. For the purpose of this example, we do not need to
describe this extension.

Now, recalling that $\Sp(H_0)=\{2n+1:n\in\Z^+\}$, we have the spectrum
of a selfadjoint extension of $H$ (hence all of them due to the
interlacing property) that does satisfy neither (C1) nor (C2) in
proposition~\ref{prop:spectrum-tells-if-operator-is-n-entire}. Therefore
$H$ is not $n$-entire for all $n\in\Z^+$.
\end{example}

%%%%%%%%%%%%%%%%%%%%%%%%%%%%%%%%%%%%%%%%%%%%%%%%%%%%%%%%%%%%%%%%%%%%%%%%%%%%%%%

\section{\texorpdfstring{On Gelfand triplets associated to
		$n$-entire operators}
		{On Gelfand triplets associated to n-entire operators}}
\label{sec:gelfand-triplets}

We recall that Krein's notion of operators entire in the generalized
sense is formulated in terms of a triplet of Hilbert spaces that
arises from the domain of the adjoint operator. In this section we aim
to obtain an analogous result for densely defined operators in any of
the classes
$\nentireclass{n}$. Our derivation however will be
more convoluted as it will based on construing the pair $\cB$ and
$\assoc_n\cB$ as part of a Gelfand triplet.
%The following discussion is rather sketchy.

As noted in remark~\ref{rem:assoc-dB-space}, the linear space $\assoc_n\cB(e)$
becomes a de Branges space when endowed (for instance) with the inner product
\begin{equation}
\label{eq:inner-assoc}
  \inner{f(x)}{g(x)}_{-n}
  		:=\int_\R\frac{\cc{f(x)}g(x)}{(x^2+1)^n\abs{e(x)}^2}dx
%  		\inner{\cdot}{\cdot}_{\mathcal{B}((z+i)^n e)},
\end{equation}
We remark that the inner product \eqref{eq:inner-assoc} is not the only
possible choice. In spite of this, we will stick to \eqref{eq:inner-assoc} in
order to simplify the ongoing discussion. Let us henceforth denote the spaces
$\assoc_n\cB(e)$ with the inner product \eqref{eq:inner-assoc} as
$\cB_{-n}$. It is clear that
\begin{equation*}
  \cB(e)=:\cB_0\subset\cB_{-1}\subset\cB_{-2}\subset\cdots
\end{equation*}
and moreover $\norm{f(x)}_{-n}\le\norm{f(x)}_{-n+1}$ for every
$f(z)\in\cB_{-n+1}$. Let $S_{-n}$ denote the operator of multiplication
with maximal domain in $\cB_{-n}$.
%We sometimes will write $S$ instead of $S_0$.

\begin{lemma}
\label{lemma:density}
Assume that $S_0$ is densely defined in $\cB_{0}$. Then
$S_{-n}$ is densely defined in $\cB_{-n}$. Also, $\cB_{0}$ is dense
in $\cB_{-n}$.
\end{lemma}
\begin{proof}
Let us start by proving the assertion for $n=1$. Consider $f(z)\in\cB_{-1}$.
Then $f(z)=h(z)+zg(z)$ for some $h(z),g(z)\in\cB_0$. Since
$S_0$ is densely defined in $\cB_0$,
there exists a sequence $\{g_l(z)\}\subset\dom(S_0)$ converging to $g(z)$ in
the $\cB_0$ norm. Define $f_l(z):=h(z)+zg_l(z)$. Then clearly
$\{f_l(z)\}\subset\cB_0$. Moreover, since
\[
\norm{xg(x)-xg_n(x)}^2_{-1}
	=   \int_\R \frac{x^2\abs{g(x)-g_n(x)}^2}{(x^2+1)\abs{e(x)}^2} dx
	\le \norm{g(x)-g_n(x)}^2_{0},
\]
$\{f_l(z)\}$ converges to $f(z)$ in the $\cB_{-1}$ norm (hence pointwise
uniformly on compact subsets).

So far, we have proven that $\cB_{0}$ is dense in $\cB_{-1}$. Since
$\cB_0\subset\dom(S_{-1})$, the latter operator is densely defined.
We now proceed by induction assuming that $S_{-n+1}$ is densely defined in
$\cB_{-n+1}$ and noting that $\cB_{-n}=\assoc\cB_{-n+1}$.

The fact that $\cB_0$ is dense in $\cB_{-n}$ follows from the ordering of
norms.
\end{proof}

%\begin{remark}
%The proof above shows that $\cB_{-m}$ is dense in $\cB_{-n}$ for every
%$0\le m<n$. %This fact will not be used in this paper.
%\end{remark}

Fix $n\in\mathbb{N}$. We aim to find a linear space
$\cB_{+n}\subset\cB_0$ such that $\{\cB_{+n},\cB_0,\cB_{-n}\}$ is a
Gelfand triplet. Most of the following discussion is based on standard
arguments; see \cite{berezanskii}.

Let $\mb{D}:\cB_{-n}\to\cB_0$ be the adjoint of the immersion map
from $\cB_0$ into $\cB_{-n}$. This linear map is well defined as
long as $\cB_0$ is dense in $\cB_{-n}$, that is, under the condition of
lemma~\ref{lemma:density}, and it turns out to be one-to-one. By definition
one has
\begin{equation}
\label{eq:inner-product-d}
\inner{g(x)}{\mb{D}f(x)}_0 = \inner{g(x)}{f(x)}_{-n},
\end{equation}
for any $f(z)\in\cB_{-n}$ and $g(z)\in\cB_0$.
Since the immersion map has norm less than one, the same holds true
for $\mb{D}$.

The map $\mb{D}$ gives an explicit relation between the
reproducing kernels $k_0(z,w)$ and $k_{-n}(z,w)$.

\begin{lemma}
\label{lem:identity-between-kernels}
For all $f(z)\in\cB_{-n}$ one has
\[
(\mb{D}f)(z) = \inner{k_0(x,z)}{f(x)}_{-n}.
\]
Moreover, $k_0(z,w) = (\mb{D}k_{-n})(z,w)$.
\end{lemma}
\begin{proof}
Since $(\mb{D}f)(z)\in\cB_0$, and taking into account
\eqref{eq:inner-product-d}, one has
\[
(\mb{D}f)(z) = \inner{k_0(x,z)}{\mb{D}f(x)}_0
		   = \inner{k_0(x,z)}{f(x)}_{-n},
\]
thus leading to the first assertion. The second assertion follows
by noticing that
\[
\inner{(\mb{D}k_{-n})(x,w)}{g(x)}_0
	= \inner{k_{-n}(x,w)}{g(x)}_{-n}
	= g(w)
	= \inner{k_0(x,w)}{g(x)}_0,
\]
for all $g(z)\in\cB_0$ (as a subset of $\cB_{-n}$) and all $w\in\C$.
\end{proof}

Now define the space $\cB_{+n}^{\pre}:=\mb{D}\cB_0$ equipped with the
sesquilinear form
\begin{equation}
\label{eq:inner-product-d-inverse}
  \inner{\cdot}{\cdot}_{+n} := \inner{\cdot}{\mb{D}^{-1}\cdot}_0\,.
\end{equation}

\begin{lemma}
The sesquilinear form $\inner{\cdot}{\cdot}_{+n}$ is an inner product so
$\cB_{+n}^{\pre}$ is an inner product space which turns out to be not complete.
Moreover, $\norm{g(x)}_{+n}>\norm{g(x)}_0$ for $g(z)\in\cB_{+n}^{\pre}$.
\end{lemma}
\begin{proof}
The first assertion follows from the fact that
\begin{eqnarray*}
\norm{g(x)}_{+n}^2
	  &=& \inner{(\mb{D}\mb{D}^{-1}g)(x)}{(\mb{D}^{-1})g(x)}_0\\[1mm]
      &=& \inner{(\mb{D}^{-1}g)(x)}{(\mb{D}^{-1}g)(x)}_{-n}
       =  \norm{(\mb{D}^{-1}g)(x)}_{-n}^2.
\end{eqnarray*}
This implies that $\mb{D}^{-1}$ has bounded norm (equal to one) as a linear
map from $\cB_{+n}^{\pre}$ to $\cB_{-n}$. Since $\cB_0$ is not closed as a
subset of $\cB_{-n}$, the second assertion follows. Finally, since
$\norm{\mb{D}}<1$, $\norm{g(x)}_0<\norm{(\mb{D}^{-1}g)(x)}_{-n}$ for all
$g(z)\in\cB_0$, thus implying the last statement.
\end{proof}

Let us denote by $\cB_{+n}$ the completion of
$\cB_{+n}^{\pre}$ with respect to the norm $\norm{\cdot}_{+n}$.
Let $\mb{T}^{\pre}$ be the restriction of
$\mb{D}^{-1}$ to $\cB_{+n}^{\pre}$. This operator can
be seen as a densely defined map in $\cB_{+n}$ with range
in $\cB_{-n}$. Now, denote by $\mb{T}$ the extension
of $\mb{T}^{\pre}$ by continuity to the whole space $\cB_{+n}$.
Since for $f(z),g(z)\in\cB_{+n}^{\pre}$,
\begin{eqnarray*}
\inner{(\mb{T}f)(x)}{(\mb{T}g)(x)}_{-n}
	&=& \inner{(\mb{D}^{-1}f)(x)}{(\mb{D}^{-1}g)(x)}_{-n}\\[1mm]
	&=& \inner{(\mb{D}^{-1}f)(x)}{g(x)}_{0}
	 = \inner{f(x)}{g(x)}_{+n}
\end{eqnarray*}
as it follows from \eqref{eq:inner-product-d} and
\eqref{eq:inner-product-d-inverse}, one has
\begin{equation*}
\inner{(\mb{T}f)(x)}{(\mb{T}g)(x)}_{-n}=\inner{f(x)}{g(x)}_{+n}
\end{equation*}
for all $f(z),g(z)\in\cB_{+n}$ due to the continuity of the inner
product.
Also, $\ran(\mb{T})=\cB_{-n}$. For suppose there is
$h(z)\in\cB_{-n}$ orthogonal to $\ran(\mb{T})$. Then one has
\begin{eqnarray*}
\inner{h(x)}{(\mb{T}g)(x)}_{-n}
	&=& \inner{h(x)}{(\mb{D}^{-1}g)(x)}_{-n}\\[1mm]
	&=& \inner{(\mb{D}h)(x)}{(\mb{D}^{-1}g)(x)}_0
	 = 0
\end{eqnarray*}
for all $g(z)\in\cB_{+n}^{\pre}$. Since $\mb{D}^{-1}\cB_{+n}^{\pre}=\cB_0$, it
follows that $(\mb{D}h)(z)\equiv 0$, thus the claim.

It follows from the construction above that the spaces $\cB_{+n}$,
$\cB_0$ and $\cB_{-n}$ form a Gelfand triplet. The duality bracket
$\dual{\cdot}{\cdot}:\cB_{+n}\times\cB_{-n}\to\C$ is given by
   \begin{equation*}
     \dual{h(x)}{f(x)}=\inner{(\mb{T}h)(x)}{f(x)}_{-n},\qquad
     h(z)\in\cB_{+n},\quad f(z)\in\cB_{-n}.
   \end{equation*}
Moreover, $\cB_{+n}$ is more than just a dense linear manifold within $\cB_0$.

\begin{proposition}
$\cB_{+n}$, equipped with the inner product $\inner{\cdot}{\cdot}_{+n}$,
is a de Branges space.
\end{proposition}
\begin{proof}
Given $w\in\C$, define $k_{+n}(z,w):=(\mb{D}k_0)(z,w)$. For
$g(z)\in\cB_{+n}^{\pre}$ one has
\[
\inner{k_{+n}(x,w)}{g(x)}_{+n}
	= \inner{(\mb{D}k_0)(x,w)}{g(x)}_{+n}
	= \inner{k_0(x,w)}{g(x)}_0
	= g(w).
\]
The continuity of the inner product implies that $k_{+n}(z,w)$ is a
reproducing kernel for $\cB_{+n}$.

Suppose that $g(z)\in\cB_{+n}$ has a non-real zero at $z=w$. Then
$g(z)\in\ran(S_{+n}-wI)$, where $S_{+n}$ is the operator of multiplication
with maximal domain in $\cB_{+n}$. Let $V_{\cc{w}w}$ denote the Cayley
transform that maps $\Ker(S_{+n}^*-\cc{w}I)$ onto $\Ker(S_{+n}^*-wI)$. By
standard results we obtain
\[
(V_{\cc{w}w}g)(z)
	=\frac{z-\cc{w}}{z-w}g(z)\in\ran(S_{+n}-\cc{w}I)
		\subset\cB_{+n}.
\]
Moreover, $\norm{(V_{\cc{w}w}g)(x)}_{+n}=\norm{g(x)}_{+n}$.

As usual, denote $f^\#(z):=\cc{f(\cc{z})}$. Notice that, for $f(z)\in\cB_{-n}$
and $h(z)\in\cB_0$,
\begin{eqnarray*}
\inner{(\mb{D}f)^\#(x)}{h(x)}_0
	 = \inner{h^\#(x)}{(\mb{D}f)(x)}_0
	&=& \inner{h^\#(x)}{f(x)}_{-n}
	\\[2mm]
	&=& \inner{f^\#(x)}{h(x)}_{-n}
	 = \inner{(\mb{D}f^\#)(x)}{h(x)}_0,
\end{eqnarray*}
thus $(\mb{D}f)^\#(z)=(\mb{D}f^\#)(z)$. Therefore,
$(\mb{D}^{-1}g)^\#(z)=(\mb{D}^{-1}g^\#)(z)$ for all $g(z)\in\cB_{+n}^{\pre}$
which, by continuity, implies $(\mb{T}g)^\#(z)=(\mb{T}g^\#)(z)$
for all $g(z)\in\cB_{+n}$. As a consequence, $g^\#(z)\in\cB_{+n}$
whenever $g(z)\in\cB_{+n}$ and $\norm{g^\#(x)}_{+n}=\norm{g(x)}_{+n}$.
\end{proof}

From the previous discussion we see that the reproducing kernels associated to
each one of the spaces $\cB_{+n}$, $\cB_0$ and $\cB_{-n}$ are related by
the identities
\begin{equation*}
     k_{+n}(z,w)=(\mb{D} k_0)(z,w)=(\mb{D}^2 k_{-n})(z,w).
\end{equation*}
However, the identity of lemma~\ref{lem:identity-between-kernels} 
can be sharpened as follows.
\begin{proposition}
\label{prop:relation-between-reprod-kernel}
  For every $w\in\mathbb{C}$,
  \begin{equation*}
    (\mb{T}^{-1}k_{-n})(\cdot,w)=k_0(\cdot,w),
  \end{equation*}
and therefore $k_0(\cdot,w)\in\mathcal{B}_{+n}$.
\end{proposition}
\begin{proof}
Take any $m(z)\in\cB_{-n}$, then
\[
	m(z)=\inner{k_{-n}(x,z)}{m(x)}_{-n}
    	=\inner{(\mb{T}\mb{T}^{-1}k_{-n})(x,z)}{m(x)}_{-n}
    	=\dual{(\mb{T}^{-1}k_{-n})(x,z)}{m(x)}.
\]
Now, if one assumes that also $m(z)\in\cB_{0}$, then, by using the fact that
\begin{equation*}
 \dual{(\mb{T}^{-1}k_{-n})(x,z)}{m(x)}
	=\inner{(\mb{T}^{-1}k_{-n})(x,z)}{m(x)}_0\,,
\end{equation*}
one concludes that
\begin{equation}
\label{eq:aux-kernel-in-n}
 m(z)=\inner{(\mb{T}^{-1}k_{-n})(x,z)}{m(x)}_0=\inner{k_0(x,z)}{m(x)}_0,
\end{equation}
since $k_0(z,w)$ is the reproducing kernel in $\cB_{0}$. Thus, the
assertion follows from the second equality in \eqref{eq:aux-kernel-in-n}
due to the arbitrary choice of $m(z)\in\cB_{0}$.
\end{proof}
\begin{remark}
\label{rem:s-adjoint-on-k}
Since $k_0(z,w)$ satisfies
\begin{equation*}
  (S^*_0k_0)(z,w)=\cc{w}k_0(z,w),
\end{equation*}
one concludes that $k_0(z,w)$ is in $\mathcal{B}_{+n}\cap\dom((S^*_0)^n)$
for every $n\in\N$.
\end{remark}

\begin{corollary}
\label{cor:def-makes-sense-dB}
Assume $S$ densely defined on $\cB$. Let $\{\cB_{+n},\cB,\cB_{-n}\}$
be the Gelfand triplet associated to $\cB$ as above, with duality
bracket $\dual{\cdot}{\cdot}$. Then $S$ is $n$-entire if and
only if there exists an entire function $m(z)\in\cB_{-n}$ such
that $\dual{k(x,z)}{m(x)}\ne 0$ for all $z\in\C$.
\end{corollary}
\begin{proof}
By definition $S$ is $n$-entire if and only if $n+1$
entire functions $m_0(z),\ldots,m_n(z)\in\cB$ can be found such that
\[
\cB = \ran(S-zI)\dot{+}
		\Span\{m_0(z)+zm_1(z)+\cdots+z^nm_n(z)\},
		\text{ for all } z\in\C.
\]
Equivalently, $S$ is $n$-entire if and only if there exists a zero-free
entire function $m(z)\in\assoc_n(\cB)$, that is
\begin{equation}
  \label{eq:zero-free-kernel-n}
\inner{k_{-n}(x,z)}{m(x)}_{-n}\ne 0, \text{ for all } z\in\C.
\end{equation}
The assertion then follows from (the proof of)
proposition~\ref{prop:relation-between-reprod-kernel}.
\end{proof}

Let us consider a densely defined operator
$A\in\ournewclass$. Associated to $A$ we have an isometry $\Phi$ that
maps $\cH$ to the de Branges space $\hat{\cH}:=\Phi\cH$. On it,
$S:=\Phi A\Phi^{-1}$ is densely defined.  Then we can construct, by
the way previously discussed, the Gelfand triplet
$\{\hat{\cH}_{+n},\hat{\cH},\hat{\cH}_{-n}\}$. Define
\[
\cH_{+n}:=\Phi^{-1}\hat{\cH}_{+n},
\]
which is a dense linear manifold within $\cH$ and itself is a Hilbert space
if equipped with the inner product
\[
\inner{\eta}{\omega}_{+n}:=\inner{\eta}{\Phi^{-1}\mb{D}^{-1}\Phi\omega},
		\quad \eta,\omega\in\cH_{+n}.
\]
It follows from remark~\ref{rem:s-adjoint-on-k} that $\xi(z)\in\cH_{+n}$
for every $z\in\C$. Now define $\cH_{-n}$ as the set of continuous linear
functionals on $\cH_{+n}$. This linear set is a Hilbert space when equipped
with the inner product
\[
\inner{\phi}{\psi}_{-n}:=\inner{\mb{G}^{-1}\phi}{\mb{G}^{-1}\psi}_{+n},
		\quad \phi,\psi\in\cH_{-n}.
\]
where $\mb{G}$ is the standard bijection from $\cH_{+n}$ onto $\cH_{-n}$
\cite{berezanskii}. These considerations along with
corollary~\ref{cor:def-makes-sense-dB} constitute the proof of the
following proposition. Here we denote the duality bracket between
$\cH_{+n}$ y $\cH_{-n}$ also by $\dual{\cdot}{\cdot}$.

\begin{proposition}
\label{prop:definition-makes-sense}
Given a densely defined operator $A\in\ournewclass$, let
$\{\cH_{+n},\cH,\cH_{-n}\}$ be the Gelfand triplet obtained as
above. Then $A\in\nentireclass{n}$ if and only if there exists
$\eta\in\cH_{-n}$ such that $ \dual{\xi(z)}{\eta}\ne 0 $ for every
$z\in\C$.
\end{proposition}

%%%%%%%%%%%%%%%%%%%%%%%%%%%%%%%%%%%%%%%%%%%%%%%%%%%%%%%%%%%%%%%%%%%%%%%%%%%%%%%

\section{Concluding remarks}
\label{sec:conclusions}

In this paper we introduce a classification of operators within the
class $\ournewclass$ of regular, closed symmetric operators on a
(necessarily) separable Hilbert space. This classification is based on
a geometric condition that generalizes a criterion due to M. G. Krein
for his definition of operators entire and entire in the generalized
sense. These new classes $\nentireclass{n}$ of $n$-entire operators
have a number of distinctive properties and there are various
characterizations apart from the geometric condition used in their
definition. Noteworthily, there is a spectral characterization of
$\nentireclass{n}$ that may be useful in several applications.  In
this respect, the theory exposed here tentatively opens up new
directions of research related with the inverse and direct spectral
analysis of operators, particularly, one-dimensional Schr\"odinger
operators.

We also have studied the $\nentireclass{n}$ class by means of associated
Gelfand triplets, following the way Krein defined and studied the
operators entire in the generalized sense. There are several aspects
of this approach (discussed in Section \ref{sec:gelfand-triplets}) that
deserve further investigation. For instance it seems insightful to
define the Gelfand triplet for an operator in $\nentireclass{n}$ in a
more intrinsic way (that is, without resorting to a functional model).
In any case, the results discussed here shed some light on the theory
of de Branges spaces and may be of interest for those studying it.

%%%%%%%%%%%%%%%%%%%%%%%%%%%%%%%%%%%%%%%%%%%%%%%%%%%%%%%%%%%%%%%%%%%%%%%%%%%%%%%

\section*{Acknowledgments}

Part of this work was done while the second author (J.\ H.\ T.)
visited IIMAS--UNAM in March 2012. He sincerely thanks them for
their kind hospitality.

This work was partially supported by CONACYT (Mexico)
through grant CB-2008-01-99100 and by CONICET (Argentina)
through grant PIP 112-200801-01741.

%%%%%%%%%%%%%%%%%%%%%%%%%%%%%%%%%%%%%%%%%%%%%%%%%%%%%%%%%%%%%%%%%%%%%%%%%%%%%%%


\section*{References}

\begin{thebibliography}{99}

\bibitem{akhiezer1}
		 Akhiezer N I 1965 {\it The Classical Moment Problem and Some
		 Related Questions in Analysis}
		 (New York: Hafner)

\bibitem{akhiezer2}
		 Akhiezer N I and Glazman I M 1993 {\it Theory of Linear Operators in
  		 Hilbert Space}
		 (New York: Dover)

\bibitem{arens}
     	 Arens R 1961 Operational calculus of linear
      	 relations
      	 {\it Pacific J. Math.} {\bf 11} 9--23

\bibitem{berezanskii}
		 Berezanski\u{\i} Ju M 1968 {\it Expansions in Eigenfunctions of
		 Selfadjoint Operators (Translations of Mathematical Monographs vol 17)}
		 (Providence: American Mathematical Society)

\bibitem{birman}
		Birman M Sh and Solomjak M Z 1987 {\it Spectral Theory of
  		Selfadjoint Operators in Hilbert Space (Mathematics and its Applications
  		[Soviet Series])}
		(Dordrecht: D. Reidel Publishing Co.)

\bibitem{debranges}
		de Branges L 1968 {\it Hilbert Spaces of Entire Functions}
		(Englewood Cliffs: Prentice-Hall)

\bibitem{dijksma}
		Dijksma A and de Snoo H S V 1974 Selfadjoint extensions of
		symmetric subspaces
		{\it Pacific J. Math.} {\bf 54} 71--100

\bibitem{eckhardt-preprint}
        Eckhardt J Schr\"odinger operators with strongly singular 
        potentials and de Branges spaces \texttt{arXiv:1105.6355}

\bibitem{MR0883643}
        Cycon H L, Froese R G, Kirsch W and Simon B 1987 \emph{Schrödinger
        operators with application to quantum mechanics and global
        geometry (Texts and Monographs in Physics)}
        (Berlin: Springer-Verlag)

\bibitem{gorbachuk}
		Gorbachuk M L and Gorbachuk V I 1997 {\it  M. G. Krein's
		Lectures on Entire Operators (Operator Theory: Advances and
		Applications vol 97)}
		(Basel: Birkh\"{a}user)

\bibitem{hassi1}
		Hassi S and de Snoo H S V 1997 One-dimensional graph perturbations
		of selfadjoint relations
		{\it Ann. Acad. Sci. Fenn. Math.} {\bf 22} 123--164

\bibitem{hassi2}
		Hassi S, de Snoo H S V and Wrinkler H 2000 Boundary-value
		problem for two-dimensional canonical systems
		{\it Integr. Equ. Oper. Theory} {\bf 36} 445--479

\bibitem{kaltenback}
		Kaltenb\"ack M and Woracek H 1999 Pontryagin spaces of
		entire functions I
		{\it Integr. Equ. Oper. Theory} {\bf 33} 34--97

\bibitem{krein-string}
        Kac I S and Krein M G 1968 \emph{On the spectral functions
        of a string} Supplement II of the Russian edition of
        Atkinson F V \emph{Discrete and continuous boundary
        problems} (Moscow: Mir) (English translation: 1974
        American Mathematical Society Translations Ser. 2
        \textbf{103} 19--102)

\bibitem{kempf1}
		Kempf A 2004 Fields with finite information density
		{\it Phys. Rev. D} {\bf 69} 124014

\bibitem{kempf2}
		Kempf A 2004 Covariant information-density cutoff in curved space-time
		{\it Phys. Rev. Lett.} {\bf 92} 221301

\bibitem{kempf3}
		Kempf A 1999 On the three short-distance structures which
		can be described by linear operators
		{\it Rep. Math. Phys.} {\bf 43} 171--177

\bibitem{krein1}
		Krein M G 1944 On Hermitian operators with defect numbers one
		{\it Dokl. Akad. Nauk SSSR} {\bf 43} 323--326 (in Russian)

\bibitem{krein2}
		Krein M G 1944 On Hermitian operators with defect numbers one: II
		{\it Dokl. Akad. Nauk SSSR} {\bf 44} 131--134 (in Russian)

\bibitem{krein3}
		Krein M G 1944 On one remarkable class of Hermitian
		operators
		{\it Dokl. Akad. Nauk SSSR} {\bf 44} 175--179 (in Russian)

\bibitem{krein-half-bounded}
        Krein M G 1947 The theory of self-adjoint extensions of
        semi-bounded Hermitian transformations and its applications: I
        {\it Mat. Sbornik N.S.} {\bf 20(62)} 431--495 (in Russian)

\bibitem{krein4}
		Krein M G 1949 Fundamental propositions of the representation theory
		of Hermitian operators with deficiency indices $(m,m)$
		{\it Ukrain. Mat. Zh.} {\bf 2} 3--66 (in Russian)

\bibitem{langer-textorius}
		Langer H and Textorius B 1977 On generalized resolvents and
		$Q$-functions of symmetric linear relations (subspaces) in
		Hilbert spaces
		{\it Pacific J. Math.} {\bf 72} 135--165

\bibitem{langer-woracek}
        Langer M and Woracek H 2002 A characterization of intermediate
        Weyl coefficients
        {\it Monatsh. Math.} {\bf 135} 137--155

\bibitem{martin-kempf1}
        Martin R T W and Kempf A 2009 Quantum uncertainty and the spectra of
        symmetric operators
        {\it Acta Appl. Math.} {\bf 106} 349--358

\bibitem{martin}
 		Martin R T W 2011 Representation of symmetric operators with
     	deficiency indices $(1,1)$ in de Branges space.
      	{\it Complex Anal. Oper. Theory} {\bf 5} 545--577

\bibitem{remling}
		Remling C 2002 Schr\"odinger operators and de Branges spaces
		{\it J. Funct. Anal.} {\bf 196} 323--294

\bibitem{RR}
		Rosenblum M and Rovnyak J 1994 {\it Topics in Hardy Classes and
  		Univalent Functions}
  		(Basel: Birkh\"{a}user)

\bibitem{I}
		Silva L O and Toloza J H 2007 Applications of Krein's theory of
		regular symmetric operators to sampling theory
		{\it J. Phys. A: Math. Theor.} {\bf 40} 9413--9426

\bibitem{II}
  		Silva L O and Toloza J H 2010 On the spectral characterization of
  		entire operators with deficiency indices $(1,1)$
  		{\it J. Math. Anal. Appl.} {\bf 367} 360--373

\bibitem{III}
  		Silva L O and Toloza J H The spectra of selfadjoint
		extensions of entire operators with deficiency indices
		$(1,1)$ \texttt{arXiv:1104.4765}

\bibitem{simon}
		Simon B 1998 The classical moment problem as a self-adjoint finite
		difference operator
		{\it Adv. Math.} {\bf 137} 82--203

\bibitem{smuljan}
  		\u{S}muljan Ju L 1971 Representation of Hermitian operators with
  		an ideal reference subspace
  		{\it Mat. Sb. (N.S.)} {\bf 85(127)} 553--562 (in Russian)

\bibitem{strauss1}
		Strauss A 2000 Functional models of regular symmetric operators
		{\it Fields Inst. Commun.} {\bf 25} 1--13

\bibitem{strauss2}
  		Strauss A 2001 {\it Functional Models of Linear Operators in: Operator
  		Theory, System Theory and Related Topics (Operator Theory: Advances and
  		Applications vol 123)}
  		(Basel: Birkh\"{a}user) pp 469--484

\bibitem{MR1711536}
        Teschl G 2000 \emph{Jacobi operators and completely integrable
        nonlinear lattices (Mathematical Surveys and Monographs 72)}
        (Providence: American Mathematical Society)

\bibitem{tur}
        Tur E A 2001 Energy Spectrum of the Hamiltonian of the
        Jaynes-Cummings Model without Rotating-Wave Approximation
        \emph{Optics and Spectroscopy} \textbf{91(6)} 899-902.

\bibitem{winkler}
		 Wrinkler H 1995 The inverse spectral problem for canonical systems
		 {\it Integr. Equ. Oper. Theory} {\bf 22} 360--374

\bibitem{woracek}
        Woracek H 2000 De Branges spaces of entire functions closed under
		forming difference quotients
		{\it Integr. Equ. Oper. Theory} {\bf 37} 238--249

\bibitem{woracek2}
		Woracek H 2011 Existence of zerofree functions $N$-associated to a
		de Branges Pontryagin space
		{\it Monatsh. Math.} {\bf 162} 453--506

\end{thebibliography}
\end{document}